\documentclass[12pt]{article}%
\usepackage{amsfonts}
\usepackage{amssymb}
\usepackage{amsmath}
\usepackage{float}
\usepackage{color}
\usepackage{graphicx}
\usepackage[round,authoryear,comma]{natbib}
\usepackage{hyperref}%
\providecommand{\U}[1]{\protect\rule{.1in}{.1in}}
\restylefloat{figure}
\newtheorem{theorem}{Theorem}

\newtheorem{definition}{Definition}

\newtheorem{remark}{Remark}

\newtheorem{example}{Example}

\let\oldexample\example
\renewcommand{\example}{\oldexample\normalfont}
\newenvironment{proof}[1][Proof]{\noindent \textbf{#1:} }{\hfill \rule{0.5em}{0.5em}}

\parskip=\medskipamount
\begin{document}

\title{Egalitarian Pooling and Sharing of Longevity Risk \\a.k.a. \\\emph{The Many Ways to Skin a Tontine Cat}}
\author{Jan Dhaene\\Actuarial Research Group, AFI\\Faculty of Business and Economics\\KU Leuven\\B-3000 Leuven, Belgium
\and Moshe A. Milevsky\\Finance Area \\Schulich School of Business \\York University \\Toronto, CANADA}
\maketitle

\begin{abstract}
There is little disagreement among insurance actuaries and financial
economists about the societal benefits of longevity-risk pooling in the form
of life annuities, defined benefit pensions, self-annuitization funds, and
even tontine schemes. Indeed, the discounted value or cost of providing an
income for life is lower -- in other words, the amount of upfront capital
required to generate a similar income stream with the same level of
statistical safety is lower -- when participants pool their financial
resources versus going it alone. Moreover, when participants' financial
circumstances and lifespans are homogenous, there is consensus on how to share
the \textquotedblleft winnings\textquotedblright\ among survivors, namely by
distributing them equally among survivors, a.k.a. a uniform rule. Alas, what
is lesser-known and much more problematic is allocating the winnings in such a
pool when participants differ in wealth (contributions) and health
(longevity), especially when the pools are relatively small in size. The same
problems arise when viewed from the dual perspective of decentralized risk
sharing (DRS). The positive correlation between health and income and the fact
that wealthier participants are likely to live longer is a growing concern
among pension and retirement policymakers. With that motivation in mind, this
paper offers a modelling framework for distributing longevity-risk pools'
income and benefits (or tontine winnings) when participants are heterogeneous.
Similar to the nascent literature on decentralized risk sharing, there are
several equally plausible arrangements for sharing benefits (a.k.a.
\textquotedblleft skinning the cat\textquotedblright) among survivors.
Moreover, the selected rule depends on the extent of social cohesion within
the longevity risk pool, ranging from solidarity and altruism to pure
individualism. In sum, actuarial science cannot really offer or guarantee
uniqueness, only a methodology.

\end{abstract}

\newpage

\section{Introduction}

\subsection{Motivation}

One of the hallmarks of a developed country is the existance of a national
pension scheme which forces all working citizens to contribute savings to a
retirement collective, which is then used to pay retirement annuities.
National pension schemes are distinct from corporate retirement plans, which
(arguably) involve more homogenous groups and whose financial generosity is
(arguably) at the discretion of employers.

For example, in a stylized national pension scheme, all workers might
contribute 10,000 (real, inflation-adjusted) euros per year in mandatory
premiums in exchange for a benefit of 27,000 (real, inflation-adjusted) euros
per year beginning at retirement age 65. Therefore, a citizen who makes these
contributions during 40 working years, for example, from age 25 until age 65,
and then lives to (and then dies exactly at) age 85, will earn an internal
rate of return of 1\% per year, in real terms.\footnote{Explanation: The
40-year FV of the 10,000 euro is 487,766 euro, which is equal to the 20-year
PV of the 27,000 euro, when the rate is equal to 1\%. And, since both cash
flows are inflation-adjusted, the implied rate is real. These numbers (2.7
multiple between benefit and contribution) do not correspond to any particular
country and are meant as an indicative example.}

This compares favourably with other real risk-free investments around the
world and is effectively guaranteed by the national government. The
paternalistic calculus is to force all citizens to participate since many
would be unlikely to do so on their own or be able to generate these
investment returns themselves.

Alas, the challenge -- and impetus for this investigation -- is what happens
to those unhealthy retirees who don't spend 20 years in retirement and do not
live to age 85 in the above example. Those who are unfortunate to live 10
years to age 75, for example, will actually earn a negative internal rate of
return of -1.6\% per year in real (inflation-adjusted) terms. Indeed,
contributing a (non-PV adjusted) total of 400,000 euros over the entire life
in exchange for only 270,000 euros is not a good investment, especially
considering they were forced to participate.

Now, in defence of this stylized national pension scheme, the conventional and
centuries-old response by pension economists and insurance actuaries is that
for every unlucky person who only lives to age 75 there is another lucky one
who lives to a ripe old age 95. They would receive 810,000 = 30 x 27,000 euro
worth of payments during retirement and thus earn an even better 2\%
inflation-adjusted return.\footnote{The FV for the 40 years (of 10,000) is
equal to the PV for the 30 years (of 27,000) when the rate is set to 2\%.}
Moreover, these defenders argue, that is the nature of longevity risk pooling.
Winners and losers are only known at the end, \emph{ex post}, but everyone
benefits from longevity pooling, \emph{ex ante}.

Unfortunately, there is a well-known and alarming body of evidence that
survival in a group of equally aged persons is not analogous to a series of
i.i.d. coin tosses. Whether due to genetics, environment or even lifestyle
choices, longevity prospects are heterogenous for individuals at the same
chronological age. The most widely cited work documenting this effect is
Chetty et al. (2016). For these less fortunate groups, there is little chance
they will live to an advanced age and benefit from the insurance aspects of
longevity risk pooling. Within society these groups have a legitimate claim
that national pension schemes aren`t fair or equitable. The most widely cited
article linking income (or wealth) to life expectancy (in the US) is the study
by Chetty et al. (2016). Other researchers have looked at non-financial
factors for the so-called \emph{stochastic longevity gap} (a.k.a. the
non-homogeneity) vis a vis the implications for pension plans, both from a
theoretical and empirical perspective. More on this later.

This problem is more than an academic exercise in probability or a theoretical
dilemma. In early 2023, a group representing Aboriginal Australians filed suit
claiming the state pension discriminates against them because their life
expectancy is much lower than non-aboriginals. Most live to their mid 70s,
while the rest of Australia live into their 80s and 90s. Although stochastic,
their internal rate of return will fall far short of the safest risk-free
alternative. The case has garnered much international attention and is pending
before their Federal Court, and various groups are expressing similar concerns
worldwide. This is the impetus for our paper; namely, it goes back to the very
first principles and asks \emph{how should longevity risk be shared?} Another
case in which this arises -- and perhaps slightly more controversial position
-- is that unhealthy males with much lower longevity prospects should (also)
be considered for similar treatment; namely, receive higher payments for the
same level of contributions.

The literature -- on the heterogeneity of longevity and the impact on pension
fairness -- is vast and growing.\footnote{See, Ayuso et al. (2017), Bravo et
al. (2023), Coppola et al. (2022), Couillard et al. (2021), Dudel and van
Raalte (2023), Finegood et al. (2021), Himmelstein et al. (2022), Kinge et al.
(2019), Li and Hyndman (2021), Lin et al. (2017), Milligan and Schirle (2021),
Mackenbach et al. (2019), Pitacco (2019), Perez-Salamero et al. (2022),
Sanzenbacher et al. (2019), Shi and Kolk (2022), Simonovitz and Lacko (2023),
Sloan et al. (2010), Strozza (2022) and Woolf et al. (2023).}

And, while a footnote list is not a proper review of the unique contribution
of every paper in the literature, their underlying messages are identical.
Namely, the most crucial empirical takeaway is the existence of an
identifiable group within society that will not live as long as the fortunate
ones. Yet, most national and corporate pension schemes are all pooled together
in one sizeable longevity-risk-sharing fund or pool.

With the motivation and background out of the way, this paper offers a
modelling framework for distributing longevity-risk pools' income and benefits
when participants are heterogeneous. Our central insight is that -- similar to
the nascent literature on decentralized risk sharing -- there are several
equally plausible rules for sharing benefits (a.k.a. \textquotedblleft
skinning the cat\textquotedblright) among survivors. Moreover, the selected
rule depends on the extent of social cohesion within the longevity risk pool,
ranging from solidarity and altruism to pure individualism. The vehicle we
choose to use for analyzing longevity risk pooling and sharing is the
one-period tontine. Our aim is to demonstrate that there are a multiplicity of
feasible arrangements for sharing gains in a one-period model and as we
progress through these models, we will draw highlights and comparisons to the
dilemma of pension equity.

\subsection{Setting the Stage}

In this subsection, we \textquotedblleft set the stage\textquotedblright\ for
our paper's main theoretical contributions by illustrating the multitude of
ways that \emph{in theory} could allocate gains and losses from longevity. A
more formal and very general model -- especially as it relates to the notation
-- will be presented in subsequent sections. For now, imagine the following
situation. A group of $n$ investors (a.k.a. retirees) \emph{pool} together
into the following one-period longevity risk-sharing scheme. They each invest
or allocate $\pi_{j}>0$ dollars at time zero into an account earning a
one-plus risk-free rate: $(1+R)\geq1$, but they face a $p_{j}>0$ probability
of surviving to the end of the period. They decide to share the total fund
among survivors, which is also known as a one-period tontine. For the sake of
a simple numerical example, we will assume $n=3$, $R\geq0$. The first retiree
invests $\pi_{1}=\$80$, the second $\pi_{2}=\$50$ and third $\pi_{3}=\$20$.
The one-period survival probabilities are: $p_{1}=20\%,p_{2}=50\%,p_{3}=80\%$,
which reflect mortality rates over a decade at old, middle and early
retirement ages. The two tables in subsections 1.2 and 1.3 summarize the in-
and outputs for the case in which $R=0$, but the example itself would apply
for any deterministic $R\geq0$.

\begin{table}[th]
\centering
\begin{tabular}
[c]{||c|c|c|c||}\hline
$\pi_{1}$ = & \$80 & $p_{1}$ = & 20\%\\\hline
$\pi_{2}$ = & \$50 & $p_{2}$ = & 50\%\\\hline
$\pi_{3}$ = & \$20 & $p_{3}$ = & 80\%\\\hline\hline
$\sum\, \pi_{k}$ & \$150.00 & \texttt{and} & $R=0\%$\\\hline
\end{tabular}
\end{table}

Clearly, investor $j=1$ has placed the most at risk because he/she faces an
$80\%$ probability of dying and losing it all \underline{and} has invested
$\$80$. Think of the function $g(\pi,p)=\pi/p$ as a theoretical measure of
\textquotedblleft money at risk\textquotedblright, although $g(\pi,p)$ could
be defined as any function that is increasing in its first argument and
decreasing in its second. Regardless of how exactly \textquotedblleft money at
risk\textquotedblright\ is measured, investor $i=3$ risks a mere $\$20$, and
faces an $80\%$ probability of surviving, so their $g(\pi,p)=\pi/p=25$. Note
that the third investor expects to survive since $p_{3}\!>\!0.5$, while the
first investor doesn't. Ergo, and perhaps even for ethical reasons, investor
$j=1$, with a $g(\pi,p)=400$, should be entitled to a larger share of the
gains if he/she happens to (get lucky) and survive. That should be obvious,
the question is how much more. There are: $2^{n}=2^{3}=8$ different scenarios,
the most vexing is the $\omega$ in which everyone dies. Now, one can set the
rules of this game in many ways -- perhaps even by offering refunds to
beneficiaries -- but we will assume that in the scenario in which everyone
dies, which has a $(1-p_{1})(1-p_{2})(1-p_{3})=8\%$ probability (assuming
independence), the $\$150$ is taxed, taken by the Government and lost to
participants. Why should the Government be entitled to a joint life insurance
policy on $n$ heads? Perhaps it's compensation for enforcing the contract in
the other $2^{n}-1=7$ scenarios. Or, it's how things work in the real world
for unclaimed money in bank accounts. Either way, that's our assumption for
$\omega_{(2^{n})}$. There are other $n=3$ scenarios that are trivial, namely
when there is only one survivor who takes the entire: $\sum\pi_{k}=\$150$,
when $R=0$, otherwise the sum is larger. This leaves $2^{n}-n-1=4$ scenarios
in which the fund must be distributed in a non-discriminatory manner. An
informal discussion with colleagues indicates \emph{the lack of any clear
consensus} on exactly how the funds should be distributed in each of those $4$
non-trivial scenarios in a manner that is perceived as
\emph{non-discriminatory}.

Now, we should note that due to the $(8\%)$ probability (assuming
independence) that everyone dies, the expected investment return to the entire
group is strictly less than $R$ when $n<\infty$. In further sections, we will
return to the implications of this fact, as it relates to matters of fairness.
Here we simply explain why. In 7 scenarios, the entire pool shares
$\$150\times(1+R)$, but in the 8th (all dead) one, the pool gets nothing.
Ergo, the pool's expected payout is 92\% times $\$150\times(1+R)$ plus 8\%
times zero, which is $\$138\times(1+R)$, and the pool's expected investment
return is: $\left(  138/150\right)  \times(1+R)-1$, which is strictly less
than $R$. Again, we will return to this matter -- and how one can fix the
expected return so that it isn't less than $R$ -- later on in the paper.

\subsection{One Possible Allocation}

As noted earlier, we sent the above query to a number of specialists (who are
noted and thanked at the end of this paper) and received a variety of replies.
Here we offer one possible way or \textquotedblleft rule\textquotedblright%
\ that can be used to distribute the $\$150$, or more generally the
$\$150(1+R)$, when $R\geq0$, in the four \emph{non-trivial} scenarios. This
solution or allocation should help set the stage for a more general discussion
later on. \begin{table}[th]
\centering
\begin{tabular}
[c]{|c|c||c|c||c|c||c|c||c|c||c|c||c|c||c|c||}\hline
$\omega_{1} $ & $W_{1}$ & $\omega_{2}$ & $W_{2}$ & $\omega_{3}$ & $W_{3}$ &
$\omega_{4}$ & $W_{4}$ & $\omega_{5}$ & $W_{5}$ & $\omega_{6}$ & $W_{6}$ &
$\omega_{7}$ & $W_{7}$ & $\omega_{8}$ & $W_{8}$\\\hline
1 & 114.29 & 0 & 0 & 0 & 0 & 1 & 141.18 & 1 & 120 & 1 & 150 & 0 & 0 & 0 &
0\\\hline
1 & 28.57 & 1 & 120 & 0 & 0 & 0 & 0 & 1 & 30 & 0 & 0 & 1 & 150 & 0 & 0\\\hline
1 & 7.14 & 1 & 30 & 1 & 150 & 1 & 8.82 & 0 & 0 & 0 & 0 & 0 & 0 & 0 &
0\\\hline\hline
8\% & ~ & 32\% & ~ & 32\% & ~ & 8\% & ~ & 2\% & ~ & 2\% & ~ & 8\% & ~ & 8\% &
~\\\hline
\end{tabular}
\end{table}

We note here that although our discussion and methodology is applicable to the
case when $R\geq0$, the numbers displayed in the table are for $R=0$. The
logic for our allocation is as follows. Start with scenario $\omega_{1}$,
where all three investors are alive, an event with 8\% probability. The first
$(\pi_{1}=\$80,p_{1}=0.20)$ investor thinks to him or herself: Had I used
$\pi_{1}$ to purchase \emph{ a pure endowment from an insurance company -- my
payout would have been:} $\pi_{1}/p_{1}\times(1+R)=\$400\times(1+R)$, assuming
a technical interest or valuation rate equal to $R$. This is an insurance
claim, but one in which a limited pool of money is available. The insurance
claim is $g(\pi,p)\times(1+R)$, where $g(\pi,p)$ is the (above noted)
\textquotedblleft money at risk\textquotedblright\ function. Likewise, the
second $(\pi_{2}=\$50,p_{2}=0.50)$ investor is entitled to an insurance claim
of: $\pi_{2}/p_{2}\times(1+R)=\$100\times(1+R)$, and the third investor
claims: $\pi_{3}/p_{3}\times(1+R)=\$25\times(1+R)$, using the same actuarial
logic. In total, for the three survivors in scenario $\omega_{1}$, the
aggregate insurance claim is $\mathbf{C}(\omega_{1}):=\sum_{k=1}^{n}(\pi
_{k}/p_{k})\times(1+R)=\sum_{k=1}^{n}g(\pi_{k},p_{k})\times(1+R)=\$525\times
(1+R)$, but alas there is only $(1+R)\sum_{k=1}^{n}\pi_{k}=(1+R)\times\$150$
available to distribute to the pool.

So, our proposed rule is to give investors the relative fraction, i.e. their
\emph{personal} insurance claim against the \emph{aggregate} insurance claim
of the available funds. The first investor claims $400\times(1+R)$ out of a
total $525\times(1+R)$, which is 76.19\%, or $(0.762)(150)\times
(1+R)=\$114.3\times(1+R)$ from the available $\$150\times(1+R)$. This is
\emph{more} than individual $i=1$ invested, but \emph{less} than his
\emph{personal} insurance claim. Algebraically this investor takes:
$(1+R)\times(\pi_{1}/p_{1})/\mathbf{C}(\omega_{1})$. The same logic gives the
second investor $100\times(1+R)$ out of $525\times(1+R)$, or 19.04\% of the
$\$150\times(1+R)$ available, which is $\$28.57\times(1+R)$, and less than the
$\pi_{2}=\$50$ invested. The third and final investor makes a personal claim
of $25\times(1+R)$ from an aggregate claim of $525\times(1+R)$, a mere 4.76\%
of the available $\$150\times(1+R)$, for a total payout of $\$7.14\times
(1+R)$. The third investor, like the second, walks away with less than
originally invested, while the (relative) winner is investor number one, who
gets more than his original $\pi_{1}=\$80$.

The same logic can be applied to the other three scenarios $\omega_{2}%
,\omega_{4},\omega_{5}$. While the \emph{personal} insurance claim for
$(\pi/p)\times(1+R)$ remains the same for each individual survivor, the
\emph{aggregate} insurance claim paid to the survivors is lower due to the
smaller number of survivors. For example, the value of $\mathbf{C}(\omega
_{2})=125\times(1+R)$, $\mathbf{C}(\omega_{4})=425\times(1+R)$ and
$\mathbf{C}(\omega_{5})=500\times(1+R)$. Again, these are the denominators for
the fractional allocation of end-of-period available funds, where the
numerator is the \emph{personal} claim $(\pi_{k}/p_{k})\times(1+R)$.

In sum, while there are many ways to \emph{skin the tontine cat} our suggested
(general) rule for $W_{(i,j)}$, which represents the payout in scenario
$\omega_{i}$ (column) and individual $j$ (row) in the above table, can be
written as:
\begin{equation}
W_{(i,j)}\;=\;(1+R)\sum_{k=1}^{n}\pi_{k}\times\left(  \frac{(\pi_{j}%
/p_{j})\times I_{(i,j)}}{\sum_{k=1}^{n}(\pi_{k}/p_{k})\times I_{(i,k)}%
}\right)  ,\;\;j=1,\ldots,n\text{ and }i=1,\ldots,2^{n}. \label{rule}%
\end{equation}
where $I_{(i,j)}$ is the (scalar) life status of the $j$'th investor in the
$i$'th scenario. Note (once again) that our numerical example assumed $R=0$,
but nothing stops us from using the same rule or allocation for the case of
$R>0$. Moreover, in the formula above, the variable $(1+R)$ in the nominator
and the denominator cancel each other out. For example, under $\omega_{4}$,
the indicator variables are: $I_{(4,1)}=1$, $I_{(4,2)}=0$, $I_{(4,3)}=1$, and
the denominator is $(80/0.2)\times(1+R)\times1+(50/0.5)\times(1+R)\times
0+(20/0.8)\times(1+R)\times1$, which is the above noted $\mathbf{C}(\omega
_{4})=\$425\times(1+R)$. Again, the quantity in brackets in equation
(\ref{rule}) is the ratio of \emph{personal} to \emph{aggregate} insurance
claim, which is then multiplied by the money available in the pool. And, when
$i=2^{n}$, which is the $\omega$ scenario in which everyone is dead, we define
$W_{(i,j)}:=0$, despite the zero in the denominator. Equation (\ref{rule}),
which will reappear in many guises and incarnations over the paper, is a
\emph{proportional} risk-sharing rule, but other rules will be proposed and
analyzed in due time.

The above $W_{(i,j)}$ formula or equation will be parsed, analyzed and refined
later on in this paper. Still, at this early stage, we should note that the
expression is not defined (and does not make any sense) for the scenario in
which all participants die. This can be easily seen since the denominator will
be zero in this case.

\subsection{The Tontine Fund}

As explained in the above example, a single period tontine fund is an
investment made by a group of people and an administrator. Participants each
invest a certain amount. At the end of the observation period, surviving
participants share the proceeds of the tontine fund, while non-surviving
participants do not receive anything. If no participant survives, the
administrator owns all the proceeds. The tontine fund is self-financing,
meaning that only the available proceeds are distributed, and insolvency or
default is not possible. Therefore, no solvency capital needs to be set up at
the beginning of the contract.

Denuit, Hieber \& Robert (2022) and Denuit \& Robert (2023) studied single
period tontine funds, also known as longevity funds or endowment contingency
funds. Indeed, the literature of tontines as well as group self-annuitzation
schemes more generally is large and continues to grow.\footnote{Key papers in
that literature include, alphabetically listed, Bernhardt \& Donelly (2019),
Bernhardt and Qu (2023), Blake, Boardman \& Cairns (2014), Chen, Chen \& Xu
(2022), Donnelly (2018), Donnelly, Guillien \& Nielsen (2014), Forman and
Sabin (2015), McKeever (2009), Piggott, Valdez \& Detzel (2005), Stamos
(2008), Weinert and Grundel (2021).}

We consider tontine funds where only the surviving participants (or the
administrator in their absence) receive the proceeds. However, it is possible
to distribute the proceeds among participants who meet objective criteria
other than survival, such as a pre-defined health event, hospitalization, etc.
The mathematical description is similar, but we focus only on survival as a
trigger for participants to be entitled to proceeds.

In this paper, we will discuss the scenario where initial investments (wealth)
and survival probabilities (health) vary among participants, which we call a
heterogeneous case. As a special case, we will also examine the situation
where all participants invest the same amount and have the same survival
probabilities, which we refer to as the homogeneous case.

As we noted earlier in this paper, one concern with (re-introducing) tontines
is their actuarial fairness. Previous research has also addressed the problem
we mentioned in our introduction. Namely, a single-period tontine where the
probability that all members die before the end date is positive, is clearly
unfair (mathematically) since the expected return for the group, after
accounting for any investment gains, is less than $R$. To resolve this issue,
some researchers have suggested adding an insurance benefit for beneficiaries
of deceased members. Indeed, most papers on tontines written in the last few
years have all added this element to repair expectations.

While this approach resolves the mathematical problem, this isn't why people
by tontines.\ Indeed, it violates the spirit of the (historical) tontine in
which all rights and ownership benefits are lost at death. This isn't why
people buy tontines. Furthermore, some members may not have any beneficiaries,
leading to yet another unintended redistribution of wealth. In extreme cases,
when there is only one person surviving, this creates a moral hazard. In other
words, and for many reasons, while adding a death benefit refund or payout
\textquotedblleft solves\textquotedblright\ the math, it \textquotedblleft
ruins\textquotedblright\ the elegance of the tontine ideal. Instead, the
approach we will present in this paper is new or at least different from the
recent literature. We introduce a \emph{tontine administrator} as both a
technical and real-world solution to some of these issues, instead of
artificially adding legacy or bequest payouts. The same administrator could be
invoked within when this problem is examined thru the prism of decentralized
risk sharing (DRS), although in that context this \textquotedblleft new
player\textquotedblright\ would serve as a legal enforcer more than a
mechanism for creating actuarial fairness. More on this DRS aspect is
discussed in the appendix.

\emph{Why an administrator?}

The (modern) tontine scheme is designed to eliminate the need for guarantees,
capital, and solvency requirements. However, to ensure that all participants
in the scheme behave appropriately, an \textquotedblleft
authority\textquotedblright\ must monitor and enforce the \textquotedblleft
rules of the game\textquotedblright. This is not just a real-world friction
but a critical aspect of the tontine scheme, as it creates the necessary legal
and administrative confidence that payouts will be shared according to
pre-specified rules. The tontine administrator, who could be a government
agency or regulator, is thus, in our view, a key participant in the scheme and
must be provided with compensation for their services. This compensation is
the extra leftovers noted above, allocated or bequeathed to the tontine
administrator when everyone dies. As we will show, if the administrator
contributes to the initial investments, this approach may make the scheme
actuarially fair and more realistic for implementation.

Indeed, one of the co-authors of this paper was involved in the introduction
of a tontine scheme in Canada and can attest to the fact that participants
were extremely concerned about who would monitor and oversee the tontine,
since the traditional insurance regulators, who demand capital, were absent.
Thus, while a utopian version of longevity risk sharing assumes everyone
behaves appropriately and discloses their true survival probabilities, we
argue that an administrator is required to keep everyone honest.

In sum, this paper introduces a new player into the (modern) tontine
literature, an administrator, and shows how they interact and engage with the
group, as well as whether or not they might be asked to pay (which means they
also contribute to the fund) for the right to administer if indeed they are
going to benefit from the tontine leftovers.

\subsection{A simple example}

Consider two individuals, denoted as person 1 and person 2, who want to
participate in a peculiar game of chance. To enter the game, person 1 pays an
amount of money denoted by $\pi_{1}$, while person 2 pays an amount of money
denoted by $\pi_{2}$. Person 1 tosses a two-sided coin, and person 2 rolls a
six-sided die. In this game, person 1 is successful if they toss heads, while
person 2 is successful if they roll a 1. Person 1 and person 2 are referred to
as "participants" in the game.

In addition to the two participants, a third person is involved, known as the
"administrator". The administrator is also allowed to contribute to the prize
pool by paying an amount of money denoted by $\pi_{3}$. According to
predetermined rules, the administrator is responsible for collecting the money
and distributing it after the coin and die are thrown.

If the coin lands on heads and the die does not land on 1, the total amount of
$\pi_{1}+\pi_{2}+\pi_{3}$ is awarded to person 1. Similarly, if the coin does
not land on heads but the die lands on 1, the total amount of $\pi_{1}+\pi
_{2}+\pi_{3}$ is awarded to person 2. If both participants are successful
(i.e., heads and 1 appear after the respective throws), the total proceeds of
$\pi_{1}+\pi_{2}+\pi_{3}$ are shared by person 1 and person 2 in a
well-defined manner. Finally, if both participants are not successful (i.e.,
neither heads nor 1 appear after their respective throws), the total proceeds
of $\pi_{1}+\pi_{2}+\pi_{3}$ go to the administrator. At first glance, it may
seem unusual that the administrator also contributes an amount of money
$\pi_{3}$ to the prize pool. However, in our example, the probability that the
administrator will receive the entire prize pool is $\frac{5}{12}$ assuming
independence, and thus it seems reasonable that the administrator should also
contribute to the prize pool for this chance of winning.

If at most one of the two participants is successful, the rules for paying out
the total proceeds are clear. However, in this paper, we seek to answer the
following question: What is a reasonable, acceptable, and possible way to
allocate the total proceeds ($\pi_{1}+\pi_{2}+\pi_{3}$) in case both
participants have a successful throw? A uniform allocation where each
participant receives $\frac{\pi_{1}+\pi_{2}+\pi_{3}}{2}$ is often seen as
'unfair' because it does not consider that the chances for success are much
larger for participant $1$ than for participant $2$. To address this, we will
denote the payouts to participants $1$ and $2$ in case of a successful throw
by $\beta_{1}$ and $\beta_{2}$, respectively, with%
\[
\beta_{1}+\beta_{2}=\pi_{1}+\pi_{2}+\pi_{3}.
\]

Introducing the indicator variables $I_{1}$ and $I_{2}$ where $I_{i}$ is 1 for
a successful participant $i$ and 0 otherwise, the payouts $W_{1}$ and $W_{2}$
can be expressed as follows:%

\[
\left(  W_{1},W_{2}\right)  =\left\{
\begin{array}
[c]{cc}%
\left(  \pi_{1}+\pi_{2}+\pi_{3},0\right)  & :\text{if }I_{1}=1\text{ and
}I_{2}=0\\
\left(  0,\pi_{1}+\pi_{2}+\pi_{3}\right)  & :\text{if }I_{1}=0\text{ and
}I_{2}=1\\
\left(  0,0\right)  & :\text{if }I_{1}=0\text{ and }I_{2}=0\\
\left(  \beta_{1},\beta_{2}\right)  & :\text{if }I_{1}=1\text{ and }I_{2}=1
\end{array}
\right.
\]
\qquad

In order to write down the payout $W_{3}$ to the administrator, we introduce
the indicator variable $I_{3}$ which is defined by%

\[
I_{3}=\left(  1-I_{1}\right)  \times\left(  1-I_{2}\right)  .
\]

This indicator variable equals $1$ when neither participant is successful
(i.e., $I_{1}=I_{2}=0$) and $0$ otherwise (i.e., $I_{1}=1$ or $I_{2}=1$). The
administrator's payoff can be expressed as:%

\[
W_{3}=\left\{
\begin{array}
[c]{cc}%
0 & :\text{if }I_{3}=0\\
\pi_{1}+\pi_{2}+\pi_{3} & :\text{if }I_{3}=1
\end{array}
\right.
\]

Before playing a gamble, the two participants and the administrator must agree
on a series of payments represented by $\pi_{1}$, $\pi_{2}$, and $\pi_{3}$, as
well as appropriate values for $\beta_{1}$ and $\beta_{2}$. In this paper, we
will explore such "exotic" gambles or investments and examine the properties
that such investments should have. We will not limit ourselves to the case of
only two participants and one administrator, but instead consider the general
problem of multiple participants and one administrator.

With some of the introductory concepts and notation behind us, the structure
of what follows in this paper is organized as follows. Section 2 models and
discusses the process of allocating tontine share. The subsequent Section 3
examines a multitutde of expressions for the payout of a tontine fund. Then,
Section 4 moves on to matters of actuarial fairness, while Section 5 links the
tontine fund to (classical) pure endowment insurance. Section 6 looks at
internal share allocation schemes and Section 7 concludes the paper. Finally,
an appendix numbered Section 8, flushes out the connection between tontine
funds more generally and decentralized risk sharing rules.

\section{Tontine funds and tontine shares}

Let's consider a scenario where a group of $n$ individuals decide to set up a
one-period tontine fund. These individuals are referred to as 'participants'.
At the beginning of the investment period, each participant $i$ makes an
initial (strictly positive) investment $\pi_{i}$ in the fund. Our objective is
to establish a fair and practical method for the surviving participants to
divide the total investment among themselves if one or more of them survives.
There is also a possibility that all participants may pass away, in which
case, we need to determine what happens to the fund's proceeds. We have an
administrator (party $n+1$) to manage the fund. The administrator's role is to
collect investments at the beginning of the investment period, invest them,
and distribute the proceeds (initial investments and returns) to the surviving
participants. If all participants pass away, the administrator receives the
full proceeds of the fund. The administrator also contributes an initial
(non-negative) investment $\pi_{n+1}$ to the fund to receive these funds in
case of no survivals. To make things simpler, we introduce a vector:%
\[
\boldsymbol{\pi}=\left(  \pi_{1},\pi_{2},\ldots,\pi_{n},\pi_{n+1}\right)  ,
\]
which we will call the \textit{investment vector}.\ 

The sum of all the investments made by the participants and the administrator,
i.e. $\sum_{j=1}^{n+1}\pi_{j}$, equals the total value of the fund at the time
$0$. Each participant invests $\pi_{i}$ to buy shares or units in the fund.
Each participant $i$ who survives until time $1$ will cash in exchange for his
shares. This paper aims to determine a reasonable and acceptable number of
units each participant should receive at time $0$ for their initial investment
of $\pi_{i}$. We will consider both the chance of inheriting part of the
tontine fund and the initial investment amount while answering this question.

Let us denote the (strictly positive) number of shares of the tontine fund
received by participant $i$ by $f_{i}$.\ The vector $\mathbf{f}$ defined by
\[
\mathbf{f}=\left(  f_{1},f_{2},\ldots,f_{n}\right)
\]
will be called the \textit{(tontine) share allocation vector}.

At time $0$, the total number of shares issued is calculated by adding up all
the shares held by participants, represented by $\sum_{j=1}^{n}f_{j}$.
Similarly, the total investment in the fund at that time is calculated by
adding up the contributions of all participants and the administrator,
represented by $\sum_{j=1}^{n+1}\pi_{j}$. It's important to note that the
administrator does not receive any shares, but in case no participant
survives, all proceeds from the fund will belong to the administrator.

We define the time-$0$ value $S(0)$\ of a tontine share as follows:%
\begin{equation}
S(0)=\frac{\sum_{j=1}^{n+1}\pi_{j}}{\sum_{j=1}^{n}f_{j}}\text{.} \label{D7}%
\end{equation}
Notice that in the denominator of (\ref{D7}), we divide by the number of
allocated shares, that is $\sum_{j=1}^{n}f_{j}$.

At an individual level, the participant's initial investment is not
necessarily equal to the time-$0$ value of his allocated tontine shares.
Indeed,
\[
\pi_{i}\neq S(0)\times f_{i}\text{, }\qquad\text{for }i=1,2,\ldots,n.
\]

It's important to understand that the symbol $\neq$ means "not necessarily
equal" here. This is because, in certain situations, two participants with the
same investment $\pi_{i}$ might require different rewards. For example,
suppose the first person has a lower survival probability due to a higher risk
profile (older age). In that case, they might need to be compensated for the
extra risk they're taking by receiving more shares than the second person. In
other situations, giving more to those with higher survival probabilities
could be more appropriate, as they are expected to live longer and will need
more financial support. We'll explore this issue further in this text.

It is also important to note that the shares or units are personalized,
meaning that they are not anonymous. Each unit sold at time $0$ is linked to a
particular individual participant in the fund. Moreover, the allocated shares
of each participant can only be exchanged by him for cash at the end of the
observation period and only if he survives it. If the participant dies during
the observation period, then his units become worthless, and we will say that
his tontine shares 'die' in that case.

The number of 'surviving' shares (i.e. shares of which the owner is still
alive at time $1$) is given by
\begin{equation}
\sum_{j=1}^{n}f_{j}\times I_{j}, \label{E1}%
\end{equation}
where $I_{j}$ stands for the indicator variable (Bernouilli r.v.) which equals
$1$ in case participant $j$ survives and equals $0$ otherwise.\ On the other
hand, the number of shares of which the owner has passed away is given by
\begin{equation}
\sum_{j=1}^{n}f_{j}\times\left(  1-I_{j}\right)  . \label{E2}%
\end{equation}
Notice that (\ref{E1}) and (\ref{E2}) may be equal to $0$ and $\sum_{j=1}%
^{n}f_{j}$, respectively, which will happen in case all participants die.\ 

Apart from the survival indicator variables related to the $n$ participants,
we also introduce an indicator variable $I_{n+1}$, that is related to the
payoff that the administrator will receive. Specifically, $I_{n+1}=1$ if all
participants die and the administrator receives all the fund's proceeds.
Conversely, $I_{n+1}=0$ if at least one participant survives and the
administrator does not receive any proceeds from the fund. Hence,
\begin{equation}
I_{n+1}=%
{\displaystyle\prod\limits_{j=1}^{n}}
\left(  1-I_{j}\right)  . \label{D7a}%
\end{equation}
We have that
\begin{equation}
I_{n+1}=1\Leftrightarrow I_{1}=I_{2}=\ldots=I_{n}=0.
\end{equation}
Hereafter, we will always assume that $0<\Pr\left[  I_{n+1}=0\right]  <1$ or,
equivalently,
\begin{equation}
0<\Pr[I_{n+1}=1]<1\text{.} \label{D35}%
\end{equation}
This assumption means that the probability that all participant die is
strictly positive and also strictly smaller than $1$.

To differentiate the shares owned by participants of the tontine fund from
regular, anonymous shares, we refer to them as 'tontine shares'. These are
individualized shares belonging to a specific person that become worthless in
the event of their death.

At time $1$, the total investment in the tontine fund has grown to
\[
\left(  1+R\right)  \times\left(  \sum_{j=1}^{n+1}\pi_{j}\right)  ,
\]
where $R$ is the return over the observation period.\ We assume that $R$ is
deterministic.\ Notice that we can generalize all coming results to the case
where $R$ is random, by replacing $R$ by $E\left[  R\right]  $ in all
formulas, provided we assume that $R$ and the $I_{i}$ are mutually
independent.\footnote{A separate proof of this can be made available by the
authors.}

As previously discussed, we calculate share allocations in a manner such that
if no participants survive, the administrator receives the entire time - $1$
value of the fund. However, if at least one participant survives, the time-1
value of the fund is distributed among the surviving participants, with each
surviving share having a value of $S(1)$ which is defined by the following
expression:%
\begin{equation}
S(1)=\left(  1+R\right)  \times\frac{\left(  \sum_{j=1}^{n+1}\pi_{j}\right)
}{\sum_{j=1}^{n}f_{j}\times I_{j}},\qquad\text{if }I_{n+1}=0. \label{D5a}%
\end{equation}
In case no participants survive, there are no surviving shares left and hence,
we don't have to define $S(1)$ in that case.\ 

Let us denote the time - $1$ payouts to the participants and the administrator
by $W_{i}$, for $i=1,2,\ldots,n+1$.\ To define these payouts, we have to
consider the cases $I_{n+1}=0$ (i.e.\ at least one participant survives) and
$I_{n+1}=1$ (i.e. not any participant survives).\ We will introduce the
notations $\left(  W_{i}\mid I_{n+1}=0\right)  $\ and $\left(  W_{i}\mid
I_{n+1}=1\right)  \ $to distinguish between these two cases.

The payouts $W_{i}$\ to the participants and the administrator are defined
hereafter.\ Conditional on $I_{n+1}=0$, i.e. at least one participant
survives, we have that the payouts to the participant and the administrator
are given by%
\begin{equation}
\left(  W_{i}\mid I_{n+1}=0\right)  =\left\{
\begin{array}
[c]{cc}%
S(1)\times f_{i}\times I_{i}, & \text{for }i=1,2,\ldots,n,\\
0, & \text{for }i=n+1.
\end{array}
\right.  \label{D5b}%
\end{equation}
Taking into account the expression (\ref{D5a}) of $S(1)$, the conditional
payouts for the participants can be expressed as follows:
\begin{equation}
\left(  W_{i}\mid I_{n+1}=0\right)  =\left(  1+R\right)  \times\left(
\sum_{j=1}^{n+1}\pi_{j}\right)  \times\frac{f_{i}}{\sum_{j=1}^{n}f_{j}\times
I_{j}}\times I_{i},\qquad\text{for }i=1,2,\ldots,n. \label{MD1'}%
\end{equation}
Hence, in case at least one participant survives, the total proceeds of the
fund, that is $\left(  1+R\right)  \times\left(  \sum_{j=1}^{n+1}\pi
_{j}\right)  $, are shared among all surviving participants, where any
survivor receives a part of the total funds, which is proportional to the
number of tontine shares $f_{i}$ which were allocated to him at the set-up
time of the fund.\ In this case, the administrator does not receive any payment.\ 

On the other hand, in case $I_{n+1}=1$, i.e. not any participant survives, the
payouts to all parties involved are defined by
\begin{equation}
\left(  W_{i}\mid I_{n+1}=1\right)  =\left\{
\begin{array}
[c]{cc}%
0, & \text{for }i=1,2,\ldots,n,\\
\left(  1+R\right)  \times\left(  \sum_{j=1}^{n+1}\pi_{j}\right)  , &
\text{for }i=n+1.
\end{array}
\right.  \label{D5c}%
\end{equation}
Hence, in case at not any participant survives, the total proceeds of the fund
are owned by the administrator, while the (heirs of the) participants do not
receive anything.\ 

From (\ref{D5b}) and (\ref{D5c}), we see that only the conditional payouts for
the participants $i=1,2,\ldots,n$, given that $I_{n+1}=0$, depend on the
number of allocated tontine shares.\ In other words, $\left(  W_{n+1}\mid
I_{n+1}=0\right)  $\ and $\left(  W_{i}\mid I_{n+1}=1\right)  $,
$i=1,2,\ldots,n+1$, are independent of the choice on the number of allocated
tontine shares.\ 

Remark that the r.v.'s $\sum_{j=1}^{n}f_{j}\times I_{j}$ and $I_{n+1}$ are
'mutually exclusive', which is a special kind of countermonotonicity,
introduced in the actuarial literature in Dhaene \& Denuit (1999). This means
that $\sum_{j=1}^{n}f_{j}\times I_{j}$ and $I_{n+1}$ are both non-negative,
while the one being strictly positive implies the other being equal to
zero.\ Hence, the realization of $\sum_{j=1}^{n+1}f_{j}\times I_{j}$, where
$f_{n+1}$ is an arbitrarily chosen strictly positive number can never be equal
to $0$. In addition to the above-noted reference, two other relevant actuarial
papers considering 'mutual exclusivity' are Cheung and Lo (2014) and Lauzier,
Lin and Wang (2024).

Taking into account this observation and the expressions (\ref{MD1'}) and
(\ref{D5c}), we can express the payouts $W_{i}$ to the participants and the
administrator in the following way:
\begin{equation}
W_{i}=\left(  1+R\right)  \times\left(  \sum_{j=1}^{n+1}\pi_{j}\right)
\times\frac{f_{i}}{\sum_{j=1}^{n+1}f_{j}\times I_{j}}\times I_{i}%
,\qquad\text{for }i=1,2,\ldots,n+1, \label{MD5}%
\end{equation}
In other words, our proposed rule is to give each surviving investor a
fraction of the available funds, where each survivor's fraction is defined as
the number of his personally appointed tontine shares against the number of
tontine shares that were appointed to all surviving participants. In case no
participants survive, then the administrator receives all available funds.
Notice that any positive value of $f_{n+1}$ is allowed, as the particular
choice does not influence the payouts $W_{i}$ of the participants and the administrator.\ 

It is a straightforward exercise to verify that the payouts to the
participants can also be written as follows:%
\[
W_{i}=\left(  1+R\right)  \times\left(  \sum_{j=1}^{n+1}\pi_{j}\right)
\times\frac{f_{i}}{f_{i}+\sum_{j\neq i}^{n}f_{j}\times I_{j}}\times
I_{i},\qquad\text{for }i=1,2,\ldots,n,
\]
where in \ $\sum_{j\neq i}^{n}f_{j}\times I_{j}$, the sum is taken over all
values $j$ from $1$ to $n$, except for $j=i$. This expression is used in
Denuit \& Robert (2023) in the special case that $\pi_{n+1}=0$ and all $f_{i}$
are equal to $1$.

From (\ref{MD5}), we find that
\begin{equation}
\sum_{j=1}^{n+1}W_{j}=\left(  1+R\right)  \sum_{j=1}^{n+1}\pi_{j}, \label{D12}%
\end{equation}
which means that the sum of all payments to the participants and the
administrator is equal to the total proceeds of the fund.\ Hence, there is no
default risk.\ For obvious reasons, we call this property (\ref{D12}) the
'\textit{self-financing property}' of the tontine fund.

Let us now introduce the notation $\mathbf{I}$ for the random vector
consisting of all the survival indicator variables of the participants:
\[
\mathbf{I}=\left(  I_{1},I_{2},\ldots,I_{n}\right)  .
\]

A tontine fund may be set up if the $n$ participants with survival indicator
vector $\mathbf{I}=\left(  I_{1},I_{2},\ldots,I_{n}\right)  $ and the
administrator agree on the vector of investments $\boldsymbol{\pi}$ and the
share allocation vector $\mathbf{f}$.\ Setting up a tontine fund only requires
a group of participants and an administrator, as well as agreement between
them on the vectors $\boldsymbol{\pi}\ $and $\mathbf{f}$. Stated differently,
the payout vector $\boldsymbol{W}=\left(  W_{1},W_{2},\ldots,W_{n+1}\right)
$, of the single period tontine fund is fully characterized by $\mathbf{I}%
,\boldsymbol{\pi}$ and $\mathbf{f}$. Therefore, we will often identify the
tontine fund with the triplet $\left(  \mathbf{I},\boldsymbol{\pi}%
,\mathbf{f}\right)  $.\ Notice that no probabilities have to be assumed to
make the tontine fund operational.\ There must only be an agreement on the
vectors $\boldsymbol{\pi}$ and $\mathbf{f}$. Of course, typically $\mathbf{f}$
may depend on $\boldsymbol{\pi}$ and eventually also on the different
participants' agreed set of survival probabilities.\ Specific choices of the
share allocation vector $\mathbf{f}$ will be considered hereafter.\ 

\section{Other expressions for the payouts of a tontine fund.}

Taking into account (\ref{D7}) we can rewrite the payouts $W_{i}$ of the
tontine fund $\left(  \mathbf{I},\boldsymbol{\pi},\mathbf{f}\right)  $ defined
in (\ref{MD5}) as follows:%
\begin{equation}
W_{i}=\left(  1+R\right)  \times\left(  S\left(  0\right)  \times\sum
_{j=1}^{n}f_{j}\right)  \times\frac{f_{i}}{\sum_{j=1}^{n+1}f_{j}\times I_{j}%
}\times I_{i},\qquad\text{for }i=1,2,\ldots,n+1.\nonumber
\end{equation}
These expressions for the payouts of the participants and the administrator
can easily be transformed into
\begin{equation}
W_{i}=S\left(  0\right)  \times\left(  1+R\right)  \times f_{i}\times\left(
1+\frac{\sum_{j=1}^{n+1}f_{j}\times\left(  1-I_{j}\right)  -f_{n+1}}%
{\sum_{j=1}^{n+1}f_{j}\times I_{j}}\right)  \times I_{i},\qquad\text{for
}i=1,2,\ldots,n+1. \label{D21'}%
\end{equation}

The expression (\ref{D21'}) of the payout $W_{i}$ to participant $i$ has a
straightforward interpretation: In case participant $i$ survives, then
$I_{n+1}=0$, and he will receive two payments at time $1$.\ The first one is
the time-$1$ value $S\left(  0\right)  \times\left(  1+R\right)  \times f_{i}$
of the tontine shares he was allocated at time $0$, where at time $1$ each
share is valuated by $S\left(  0\right)  \times\left(  1+R\right)  $.\ This
means that the first payment is the time-$1$ value of his allocated shares, in
case they were part of a financial fund with a return $R$.\ In addition, the
shares of the persons who did not survive, that is $\sum_{j=1}^{n}f_{j}%
\times\left(  1-I_{j}\right)  $, are distributed among the survivors, where
each survivor gets a part of it proportional to the shares he was allocated at
time $0$. Hence, person $i$ receives in addition $\sum_{j=1}^{n}f_{j}%
\times\left(  1-I_{j}\right)  \times\frac{f_{i}}{\sum_{j=1}^{n}f_{j}\times
I_{j}}$ shares, where each additional share is also evaluated by $S\left(
0\right)  \times\left(  1+R\right)  $. The value of these extra shares
corresponds to the second payment at time $1$.\ 

Notice that (\ref{D21'}) implies that
\begin{equation}
W_{i}\geq S\left(  0\right)  \times\left(  1+R\right)  \times f_{i}\times
I_{i},\qquad\text{for }i=1,2,\ldots,n, \label{D1}%
\end{equation}
which means that in case participant $i$ survives, he will always receive at
least the time-$1$ value of the tontine shares that were allocated to him at
time $0$, where accumulation is performed with the tontine fund return $R$.
Notice that (\ref{D1}) does not mean that upon survival, the participant
receives at least the accumulated value of his initial investment $\pi_{i}$.
Hence, upon survival,%
\begin{equation}
W_{i}\ngeqq\pi_{i}\times\left(  1+R\right)  \times I_{i},\qquad\text{for
}i=1,2,\ldots,n, \label{D2}%
\end{equation}
where $\ngeqq$ has to be interpreted as 'not necessarily larger than or equal
to'.\ Remark that $W_{i}\geq\pi_{i}\times\left(  1+R\right)  \times I_{i}$
will hold for each participant in case $f_{i}=\pi_{i}$ for all participants
$i$. We will come back to this tontine share allocation in a further section.\ 

For any $i$, we can rewrite formula (\ref{D21'}) as follows:
\begin{equation}
W_{i}=\pi_{i}\times\left(  1+R\right)  \times\left(  1+R_{i}^{\prime}\right)
\times\left(  1+R^{\prime\prime}\right)  \times I_{i},\qquad\text{for
}i=1,2,\ldots,n+1, \label{D2'}%
\end{equation}

with
\[
\left(  1+R_{i}^{\prime}\right)  =\frac{S\left(  0\right)  \times f_{i}}%
{\pi_{i}}%
\]
and
\[
\left(  1+R^{\prime\prime}\right)  =\left(  1+\frac{\sum_{j=1}^{n+1}%
f_{j}\times\left(  1-I_{j}\right)  -f_{n+1}}{\sum_{j=1}^{n+1}f_{j}\times
I_{j}}\right)  .
\]
This means that the return that participant $i$ receives on his initial
investment $\pi_{i}$ upon survival is composed of 3 parts: the investment
return $R$ of the fund, a risk adjustment return $R_{i}^{\prime}$ (because at
time $0$, the investment $\pi_{i}$ is used to buy shares, where the number of
allocated shares to each participant in one way or another reflect his risk
profile), and finally the return $R^{\prime\prime}$ which is caused by the
mortality credits, as the investments of the participants who died are shared
among the surviving participants.\ Notice that $R$ and $R^{\prime\prime}$ are
non-negative and independent of $i$, whereas the risk adjustment return
$R_{i}^{\prime}$ is participant-specific and may be negative. We should point
out that $R$ and $R_{i}^{\prime}$ are deterministic, whereas $R^{\prime\prime
}$ is stochastic.

\begin{remark}
Suppose that $\pi_{1}=\pi_{2}=\ldots=\pi_{n}=\pi$, and also that the
participants are ordered in such a way that
\[
f_{1}\leq f_{2}\leq\ldots\leq f_{n}.
\]
Intuitively, participant $1$ is the one who gets the least amount of shares
(e.g. because he is the youngest, implying that his investment is least at
risk), while participant $n$ is the one who gets the most shares (e.g. because
he is the oldest participant).\ Then we find that
\[
S\left(  0\right)  \times f_{1}=\frac{n\times\pi+\pi_{n+1}}{\sum_{j=1}%
^{n}f_{j}}\times f_{1}\leq\frac{n\times\pi+\pi_{n+1}}{n\times f_{1}}\times
f_{1}=\pi+\frac{\pi_{n+1}}{n}.
\]
This observation leads to
\[
\left(  1+R_{1}^{\prime}\right)  =\frac{S\left(  0\right)  \times f_{1}}{\pi
}\leq1+\frac{1}{n}\frac{\pi_{n+1}}{\pi}.
\]
In case $\pi_{n+1}=0$, the person who is allocated the least amount of tontine
shares person receives a negative adjustment return $R_{1}^{\prime}\leq0$. On
the other hand, one has that
\[
S\left(  0\right)  \times f_{n}=\frac{n\times\pi+\pi_{n+1}}{\sum_{j=1}%
^{n}f_{j}}\times f_{n}\geq\frac{n\times\pi+\pi_{n+1}}{n\times f_{n}}\times
f_{n}=\pi+\frac{\pi_{n+1}}{n}%
\]
and hence,
\[
\left(  1+R_{n}^{\prime}\right)  =\frac{S\left(  0\right)  \times f_{n}}{\pi
}\geq1+\frac{1}{n}\frac{\pi_{n+1}}{\pi}\geq1.
\]
This means that the 'person who gets the most shares' receives a positive risk
adjustment return $R_{n}^{\prime}\geq0$.\ 
\end{remark}

\section{Actuarial fairness of a tontine fund}

Following, Bernard, Feliciangeli \& Vanduffel (2022) and others in the next
definition, we say that a tontine fund is 'actuarially fair' for the
participants in case it is an actuarial fair deal for each participant. This
means that the time $1$ value of each participant's initial investment
$\pi_{i}$ is equal to his expected payoff $E\left[  W_{i}\right]  $ at time
$1$.

\begin{definition}
The tontine fund $\left(  \mathbf{I},\boldsymbol{\pi},\mathbf{f}\right)  $ is
actuarial fair for each of its participants in case%
\begin{equation}
\pi_{i}\times\left(  1+R\right)  =E\left[  W_{i}\right]  ,\qquad\text{for
}i=1,2,\ldots,n\text{.} \label{E6}%
\end{equation}

\end{definition}

We are loath to introduce yet additional notations or generalizations for the
precise magnitude of $R$, in a real-world scenario. But, in terms of
structure, one could differentiate between a technical (valuation) interest
rate, which is used to \emph{discount} expected cash flows, and the
deterministic return of the fund itself, that is, the rate by which the fund
grows. Carrying those two $R$s wouldn't add much to the underlying longevity
risk-sharing insights and would (only) add clutter to the equations.

As the payouts for the participants are $0$ in case not any person survives,
we have the tontine fund $\left(  \mathbf{I},\boldsymbol{\pi},\mathbf{f}%
\right)  $ is actuarially fair in the case
\begin{equation}
\pi_{i}\times\left(  1+R\right)  =E\left[  W_{i}\mid I_{n+1}=0\right]
\times\Pr\left[  I_{n+1}=0\right]  ,\qquad\text{for }i=1,2,\ldots,n.
\label{E6a}%
\end{equation}
Taking into account (\ref{MD1'}), the $n$ fairness conditions (\ref{E6a}) can
be written as follows:%
\begin{equation}
\pi_{i}=\left(  \sum_{j=1}^{n+1}\pi_{j}\right)  \times E\left[  \frac
{f_{i}\times I_{i}}{\sum_{j=1}^{n}f_{j}\times I_{j}}\mid I_{n+1}=0\right]
\times\Pr\left[  I_{n+1}=0\right]  ,\qquad\text{for }i=1,2,\ldots,n.
\label{E7}%
\end{equation}

We leave for future work, or perhaps to an enterprising student, a formal
proof that -- under some appropriate and suitable conditions -- \emph{at
least} one solution $\boldsymbol{\pi}$ exists to the above set of equations.
Also, on the topic of future work, in the event the survival probabilities are
themselves stochastic (or entirely unknown) one could devise an \emph{ex ante}
agreement for sharing the proceeds of the fund, notwithstanding the fact it
might not be \textquotedblleft actuarially fair\textquotedblright.

\begin{theorem}
If the tontine fund $\left(  \mathbf{I},\boldsymbol{\pi},\mathbf{f}\right)  $
is actuarial fair for each of its participants, then for any $\alpha>0$ and
$\beta>0$ also the tontine fund $\left(  \mathbf{I},\alpha\times
\boldsymbol{\pi},\beta\times\mathbf{f}\right)  $ is actuarially fair for all participants.
\end{theorem}

\begin{proof}
The proof follows immediately from the fairness conditions (\ref{E7}).
\end{proof}

The theorem above implies that if a tontine fund $\left(  \mathbf{I}%
,\boldsymbol{\pi},\mathbf{f}\right)  $ is actuarially fair for its
participants, then the tontine fund $\left(  \mathbf{I},\alpha\times
\boldsymbol{\pi},\mathbf{f}\right)  $, where we multiply all the investments
of all participants and the administrator by a uniform factor $\alpha$, is
also actuarially fair for these same participants. In other words, for a given
group of participants with given survival index vector $\mathbf{I}$ and given
tontine share allocation vector $\mathbf{f}$ not depending on $\boldsymbol{\pi
}$, the set of $n$ equations (\ref{E7}) with unknown $\boldsymbol{\pi}$ can
never have a single solution: if $\boldsymbol{\pi}$ is a solution of
(\ref{E7}), then for any $\alpha>0$ also $\mathbf{\alpha}\times\boldsymbol{\pi
}$ is a solution.

\begin{definition}
The tontine fund $\left(  \mathbf{I},\boldsymbol{\pi},\mathbf{f}\right)  $ is
actuarial fair for the administrator in case
\begin{equation}
\pi_{n+1}\times\left(  1+R\right)  =E\left[  W_{n+1}\right]  . \label{E7a}%
\end{equation}

\end{definition}

Taking into account that the payout to the administrator is $0$ in case at
least one participant survives, we find that the tontine fund is actuarially
fair for the administrator in case
\[
\pi_{n+1}\times\left(  1+R\right)  =E\left[  W_{n+1}\mid I_{n+1}=1\right]
\times\Pr\left[  I_{n+1}=1\right]  ,
\]
or equivalently, taking into account (\ref{D5c}), the tontine fund is
actuarial fair for the administrator if and only if%
\begin{equation}
\pi_{n+1}=\left(  \sum_{j=1}^{n}\pi_{j}\right)  \times\frac{\Pr\left[
I_{n+1}=1\right]  }{\Pr\left[  I_{n+1}=0\right]  }. \label{E8}%
\end{equation}

Notice that in case the number of participants $n$ is large, we will typically
have that the probability that at least one participant survives, i.e.
$\Pr\left[  I_{n+1}=0\right]  $, will be close to $1$. That means that in this
case, we will have that
\[
\pi_{n+1}\approx0.
\]

In the special case that all $I_{i}$ are i.i.d. with $\Pr\left[
I_{i}=0\right]  =q$, we have that $\Pr\left[  I_{n+1}=1\right]  =q^{n}$ and
$\Pr\left[  I_{n+1}=0\right]  =1-q^{n}$, and (\ref{E8}) transforms into
\[
\pi_{n+1}=\left(  \sum_{j=1}^{n}\pi_{j}\right)  \times\frac{q^{n}}{1-q^{n}}.
\]

A question that we will consider in the following theorem is whether a tontine
fund which is fair for all participants is also fair for the administrator.\ 

\begin{theorem}
A tontine fund $\left(  \mathbf{I},\boldsymbol{\pi},\mathbf{f}\right)  $ that
is actuarial fair for each of its participants is also actuarial fair for the
administrator, i.e. the conditions (\ref{E6}) imply that $\pi_{n+1}$ is given
by (\ref{E8}).
\end{theorem}

\begin{proof}
Suppose that the tontine fund is actuarial fair for each participant.\ This
means that the conditions (\ref{E7}) hold for all participants.\ Summing these
$n$ actuarial fairness conditions,
\[
\sum_{j=1}^{n}\pi_{j}=\left(  \sum_{j=1}^{n+1}\pi_{j}\right)  \times\Pr\left[
I_{n+1}=0\right]  ,
\]
implies that $\pi_{n+1}$ is given by (\ref{E8}), and hence, the tontine fund
is actuarial fair for the administrator.
\end{proof}

From the theorem above, we conclude that a necessary condition for a tontine
fund to be actuarially fair for each of its participants is that it is
actuarial fair for the administrator.\ In other words, in case a tontine fund
is not actuarial fair for its administrator, it cannot be actuarial fair for
all its participants.\ 

In the literature, usually the investment $\pi_{n+1}$ of the administrator is
set equal to $0$, which means that the tontine fund is not actuarial fair for
the administrator, which in turn implies that it can also not be actuarial
fair for all its participants. This observation can also be seen as
follows.\ In case $\pi_{n+1}=0$, we find from (\ref{MD5}) that
\begin{align*}
\sum_{j=1}^{n}E\left[  W_{j}\right]   &  =\sum_{j=1}^{n}E\left[  W_{j}\mid
I_{n+1}=0\right]  \times\Pr\left[  I_{n+1}=0\right] \\
&  =\left(  1+R\right)  \left(  \sum_{j=1}^{n}\pi_{j}\right)  \times\Pr\left[
I_{n+1}=0\right] \\
&  <\left(  1+R\right)  \left(  \sum_{j=1}^{n}\pi_{j}\right)
\end{align*}
This inequality implies that it is impossible that the tontine fund is
actuarial fair for each participant, i.e. $E\left[  W_{i}\right]  =\left(
1+R\right)  \times\pi_{i}$ for all $i$, and\ for at least one participant $i$,
one must have that $E\left[  W_{i}\right]  <\pi_{i}\times\left(  1+R\right)
$.\ Milevsky and Salisbury (2016) call such rules 'equitable' in the sense
that no specific or identifiable member is disadvantaged in time-zero (a.k.a.
initial) expectations.

More specifically, in that paper they investigate how to construct a multi-age
tontine scheme and \textquotedblleft determine the proper share prices to
charge participants so that it is equitable and doesn't discriminate against
any age or any group.\textquotedblright\ The tontine they propose is a closed
pool that does not allow anyone to enter or exit after the initial set-up. To
quote from Milevsky \& Salisbury (2016):

\begin{quote}
\textquotedblleft...By the word fair, we mean that the expected present value
of income will always be less than the amount contributed or invested into the
tontine. However, a heterogeneous tontine scheme can often (though not always)
be made equitable by ensuring that the present value of income (although less
than the amount contributed) is the same for all participants in the scheme
regardless of age. This scheme will not discriminate against any one cohort
although it won't be fair...\textquotedblright
\end{quote}

We should note that they (too) discuss the challenges in designing
longevity-risk sharing rules that work for small groups, and they conclude:

\begin{quote}
``...We have proved that it is possible to mix cohorts without discriminating
provided the diversity of the pool satisfies certain dispersion conditions and
we propose a specific design that appears to work well in practice...''
\end{quote}

Once again, this is consistent with the main tenor of this paper, that there
are an assortment or multitude of methods in which longevity risk can be
shared -- the many ways to \emph{skin a cat} -- and that \emph{a priori} one
isn't necessarily better or worse than the other.

\begin{theorem}
The tontine fund $\left(  \mathbf{I},\boldsymbol{\pi},\mathbf{f}\right)  $ is
actuarial fair for each participant if and only if the following conditions
are satisfied:%
\begin{equation}
\frac{\pi_{i}}{\pi_{n+1}}=E\left[  \frac{f_{i}\times I_{i}}{\sum_{j=1}%
^{n}f_{j}\times I_{j}}\mid I_{n+1}=0\right]  \times\frac{\Pr\left[
I_{n+1}=0\right]  }{\Pr\left[  I_{n+1}=1\right]  },\qquad\text{for
}i=1,2,\ldots,n. \label{F2}%
\end{equation}

\end{theorem}

\begin{proof}
(a) Let us first assume that the tontine fund $\left(  \mathbf{I}%
,\boldsymbol{\pi},\mathbf{f}\right)  $\ is actuarial fair for each
participant. Then we have from Theorem 2 that the tontine fund is also
actuarial fair for the administrator and his investment\ $\pi_{n+1}$ follows
from (\ref{E8}). The actuarial fairness for the participants means that
(\ref{E7}) holds for any $i=1,2,\ldots,n$.\ These $n$ expressions lead to the
stated expressions (\ref{F2}) for the participant's investments. (b) Next, we
assume the conditions (\ref{F2}) are satisfied.\ Summing these $n$ equations
leads to (\ref{E8}), which is the actuarial fairness condition for the
administrator. The conditions (\ref{F2}) can be rewritten as
\[
\frac{\pi_{i}}{\pi_{n+1}}=\left(  \frac{\Pr\left[  I_{n+1}=0\right]  }%
{\Pr\left[  I_{n+1}=1\right]  }+1\right)  \times E\left[  \frac{f_{i}\times
I_{i}}{\sum_{j=1}^{n}f_{j}\times I_{j}}\mid I_{n+1}=0\right]  \times\Pr\left[
I_{n+1}=0\right]  .
\]
Taking into account the expression (\ref{E8}) for $\pi_{n+1}$ leads to the
actuarial fairness conditions (\ref{E7}) for the participants.
\end{proof}

Those looking for an application to the above theorem might consider the
following. If a group wanted to construct a tontine fund or scheme that was
actuarially fair, the order of operations would start by choosing the
administrator's investment $\pi_{n+1}$ and then the share allocation vector
$\mathbf{f}$. The individual investments -- again, so that the scheme is
actuarially fair, would follow from (\ref{F2}). Of course, in practice, this
order is often reversed when the investment vector $\boldsymbol{\pi}$ is
chosen first, and the share allocation vector $\mathbf{f}$ is an afterthought,
depending on the choice of $\mathbf{\pi}$.

From Theorem 2, we know that if a tontine fund is actuarially fair for each of
its participants, then it is also fair for the administrator.\ However, the
opposite implication does not hold: Actuarial fairness for the administrator
is not sufficient to have actuarial fairness for each of its participants. Let
us now introduce a weaker form of actuarial fairness, which we baptize
\textquotedblleft collective actuarial fairness\textquotedblright, equally
described as \textquotedblleft socially just\textquotedblright\ to avoid the
overused and rather loaded term, fair.

\begin{definition}
The tontine fund $\left(  \mathbf{I},\boldsymbol{\pi},\mathbf{f}\right)  $ is
collective actuarial fair for its participants in case%
\begin{equation}
\left(  \sum_{j=1}^{n}\pi_{j}\right)  \times\left(  1+R\right)  =E\left[
\sum_{j=1}^{n}W_{j}\right]  \text{.} \label{E8a}%
\end{equation}

\end{definition}

Collective actuarial fairness (a.k.a. socially just) means that the time $1$
value of the sum of all participant's initial investments $\sum_{j=1}^{n}%
\pi_{j}$ is equal to the expected value of the sum of all their payoffs
$\sum_{j=1}^{n}W_{j}$ at time $1$.\ 

\begin{theorem}
A tontine fund $\left(  \mathbf{I},\boldsymbol{\pi},\mathbf{f}\right)  $ is
collective actuarial fair for its participants if and only if it is actuarial
fair for the administrator, i.e. the conditions (\ref{E8a}) and (\ref{E8}) are equivalent.
\end{theorem}

\begin{proof}
In the general case, where $\pi_{n+1}\geq0$, the expected value of the total
payouts to all participants is given by
\begin{align}
E\left[  \sum_{j=1}^{n}W_{j}\right]   &  =E\left[  \sum_{j=1}^{n}W_{j}\mid
I_{n+1}=0\right]  \times\Pr\left[  I_{n+1}=0\right] \nonumber\\
&  =\left(  1+R\right)  \times\left(  \sum_{j=1}^{n+1}\pi_{j}\right)
\times\Pr\left[  I_{n+1}=0\right]  . \label{E5}%
\end{align}
This means that the condition (\ref{E8a}) for collective actuarial fairness
can be rewritten as follows:%
\[
\left(  \sum_{j=1}^{n}\pi_{j}\right)  =\left(  \sum_{j=1}^{n+1}\pi_{j}\right)
\times\Pr\left[  I_{n+1}=0\right]  ,
\]
which is equivalent with the condition (\ref{E8}) of actuarial fairness for
the administrator.
\end{proof}

Remark that the proof of the previous Theorem also follows directly from the
self-financing property (\ref{D12}). Indeed, this property implies that
\[
E\left[  \sum_{j=1}^{n+1}W_{j}\right]  =\left(  1+R\right)  \sum_{j=1}%
^{n+1}\pi_{j},
\]
which immediately leads to the proof of the Theorem.\ 

In the following theorem, we consider the situation where the participants are
indistinguishable in the sense that the random vector $\left(  I_{1}%
,I_{2},\ldots,I_{n}\right)  $ is exchangeable.\ A special case of the
exchangeability assumption is that all $I_{i}$ are i.i.d.

\begin{theorem}
Consider the tontine fund denoted by $\left(  \mathbf{I},\boldsymbol{\pi
},\mathbf{f}\right)  $.\ Suppose that the indicator vector $\mathbf{I}=\left(
I_{1},I_{2},\ldots,I_{n}\right)  $\ is exchangeable and assume that the fund
applies a uniform tontine share allocation vector $\mathbf{f}=\left(
f,f,\ldots,f\right)  $. Then we have that the tontine fund is actuarial fair
for any of its participants if and only if the following condition is
satisfied: All participants pay the same initial investment, that is $\pi
_{i}=\pi$, for $i=1,2,\ldots,n$, with:
\begin{equation}
\pi=\frac{\pi_{n+1}}{n}\times\frac{\Pr\left[  I_{n+1}=0\right]  }{\Pr\left[
I_{n+1}=1\right]  } \label{G5}%
\end{equation}

\end{theorem}

\begin{proof}
Consider the tontine fund $\left(  \mathbf{I},\boldsymbol{\pi},\mathbf{f}%
\right)  $, where $\mathbf{I}=\left(  I_{1},I_{2},\ldots,I_{n}\right)  $\ is
exchangeable and $\mathbf{f}=\left(  f,f,\ldots,f\right)  $. From (\ref{F2}),
we know that the tontine fund is actuarial fair for any of its participants if
and only if
\begin{equation}
\pi_{i}=\pi_{n+1}\times E\left[  \frac{I_{i}}{\sum_{j=1}^{n}I_{j}}\mid
I_{n+1}=0\right]  \times\frac{\Pr\left[  I_{n+1}=0\right]  }{\Pr\left[
I_{n+1}=1\right]  },\qquad\text{for }i=1,2,\ldots,n. \label{G1a}%
\end{equation}
Taking into account the exchangeability of $\left(  I_{1},I_{2},\ldots
,I_{n}\right)  $, a symmetry argument leads to the conclusion that $E\left[
\frac{I_{i}}{\sum_{j=1}^{n}I_{j}}\mid I_{n+1}=0\right]  $ is equal for all
$i$.\ Further, as
\[
\sum_{i=1}^{n}E\left[  \frac{I_{i}}{\sum_{j=1}^{n}I_{j}}\mid I_{n+1}=0\right]
=1,
\]
we find that
\[
E\left[  \frac{I_{i}}{\sum_{j=1}^{n}I_{j}}\mid I_{n+1}=0\right]  =\frac{1}%
{n},\qquad\text{for }i=1,2,\ldots,n.
\]
We can conclude that the $n$ actuarial fairness conditions for the
participants are equivalent with $\pi_{i}=\pi$, for $i=1,2,\ldots,n$, where
$\pi$ satisfies (\ref{G5}).
\end{proof}

\section{Single period tontine vs. classical pure endowment}

Consider $n$ persons with survival indicator vector $\mathbf{I}$, who want to
set up a one-period tontine fund and start negotiations about how much
everyone should invest and how the tontine shares should be allocated.\ To
come up with a reasonable tontine fund structure characterized by $\left(
\mathbf{I},\boldsymbol{\pi},\mathbf{f}\right)  $, they start by deciding on
the vector $\boldsymbol{\pi}$.\ Once this vector is specified, the
participants observe the insurance market to find out what kind of pure
endowment insurance each could buy for a premium equal to his tontine fund
investment.\ Suppose that person $i$ can buy a pure endowment with survival
benefit $L_{i}$ for a premium equal to $\pi_{i}$. We do not require any
particular premium principle to determine the $\pi_{i}$. In other words, we
assume the $\pi_{i}$ to be chosen and the corresponding $L_{i}$ to be observed
in the market.\ 

In case the $n$ persons buy the insurance from a particular insurer, this
insurer faces a possibility of insolvency, that is a possibility that the
event
\[
\sum_{j=1}^{n}L_{j}\times I_{j}-\left(  1+R\right)  \times\sum_{j=1}^{n}%
\pi_{j}>0\text{ }%
\]
might occur.\ 

In traditional life insurance, the insurer \textquotedblright
solves\textquotedblright\ the insolvency issue by charging sufficiently high
premiums and setting up a solvency capital.\ 

To solve this issue for the tontine fund under construction, the $n$ persons
appoint an external administrator, who is assumed to contribute $\pi_{n+1}%
\geq0$.\ \ As before, we introduce the Bernoulli r.v., defined as follows:%

\begin{equation}
I_{n+1} = \prod_{j=1}^{n} \left(  1-I_{j} \right)  .
\end{equation}

Furthermore, let $L_{n+1}$ be an arbitrarily chosen strictly positive number.
Then, for each participant, the 'insurance payout' $L_{i}\times I_{i}$ is
replaced by the 'tontine payout'
\begin{equation}
W_{i}=\alpha\left(  \mathbf{I}\right)  \times L_{i}\times I_{i},\qquad
\text{for }i=1,2,\ldots,n+1, \label{D7d}%
\end{equation}
where $\alpha\left(  \mathbf{I}\right)  $ follows from
\[
\alpha\left(  \mathbf{I}\right)  \times\left(  \sum_{j=1}^{n+1}L_{j}\times
I_{j}\right)  -\left(  1+R\right)  \times\left(  \sum_{j=1}^{n+1}\pi
_{j}\right)  =0,
\]
or, equivalently,
\begin{equation}
\alpha\left(  \mathbf{I}\right)  =\left(  1+R\right)  \times\frac{\sum
_{j=1}^{n+1}\pi_{j}}{\sum_{j=1}^{n+1}L_{j}\times I_{j}}. \label{D7c}%
\end{equation}

Hence, the random coefficient $\alpha\left(  \mathbf{I}\right)  $ is chosen
such that the benefits $\alpha\left(  \mathbf{I}\right)  \times L_{i}\times
I_{i}$ satisfy the full allocation condition.

Notice that $\alpha\left(  \mathbf{I}\right)  $ is identical for any
particular participant and the administrator, but it is only observable at
time $1$. It is straightforward to verify that the particular choice of
$L_{n+1}$ does not influence the payouts $W_{i}$.

Furthermore, from (\ref{D7c}) and (\ref{D7d}) , we find that
\begin{equation}
W_{i}=\left(  1+R\right)  \times\left(  \sum_{j=1}^{n+1}\pi_{j}\right)
\times\frac{L_{i}}{\sum_{j=1}^{n+1}L_{j}\times I_{j}}\times I_{i}%
,\qquad\text{for }i=1,2,\ldots,n+1. \label{T3}%
\end{equation}

We find that these payouts are exactly equal to the payouts of the tontine
fund $\left(  \mathbf{I},\boldsymbol{\pi},\mathbf{f}\right)  $ with payouts
$W_{i}$ defined in (\ref{MD5}), provided the allocated shares are given by
\[
\mathbf{f}=\mathbf{L},
\]
where
\[
\mathbf{L}=\left(  L_{1},L_{2},\ldots,L_{n}\right)  \text{. }%
\]
This rule pays the survivors the relative fraction, i.e.\ their personal
insurance claim against the aggregate insurance claim of the surviviors, of
the available funds.\ 

In order to be able to apply (\ref{T3}), the participants and the
administrator only have to agree on the vectors $\boldsymbol{\pi}$ and
$\mathbf{L}$. This means that they only have to decide and agree on what
everyone invests at time $0$ and on what the participants would receive as
survival benefit in a classical pure endowment insurance environment for their
investment used as a premium. The choice of the premium principle or a
mortality table is not required.\ 

So far, we did not consider the choice of $\pi_{n+1}$.\ A possible choice for
the administrator's contribution is given by (\ref{E8}),
\[
\pi_{n+1}=\left(  \sum_{j=1}^{n}\pi_{j}\right)  \times\frac{\Pr\left[
I_{n+1}=1\right]  }{\Pr\left[  I_{n+1}=0\right]  }.
\]
This choice of $\pi_{n+1}$ makes the tontine fund $\left(  \mathbf{I}%
,\boldsymbol{\pi},\mathbf{L}\right)  $ fair for the administrator, and hence,
also collective fair for the group of participants.\ 

A possible way to fix $\mathbf{L}$ is chosing the $L_{i}$ such that
\[
\pi_{i}=\frac{1}{\left(  1+R^{\prime}\right)  }\times L_{i}\times p_{i}%
,\qquad\text{for }i=1,2,\ldots,n
\]
for a given (agreed) life table and technical interest $R^{\prime}$. This
means that the participants agree on a lifetable and choose the amounts
$L_{i}$ as the survival benefit corresponding to the net premium in a pure
endowment insurance with net premium $\pi_{i}$. Under this choice, we find
that (\ref{T3}) reduces to%
\begin{equation}
W_{i}=\left(  1+R\right)  \times\left(  \sum_{j=1}^{n+1}\pi_{j}\right)
\times\frac{\frac{\pi_{i}}{p_{i}}}{\sum_{j=1}^{n+1}\frac{\pi_{j}}{p_{j}}\times
I_{j}}\times I_{i},\qquad\text{for }i=1,2,\ldots,n+1. \label{T3'}%
\end{equation}

In the following section, we will come back to the particular payout scheme
defined in (\ref{T3'}).\ 

\section{Tontine funds with an internal share allocation scheme.\ }


Let us consider a group of $n$ persons with survival indicator vector
$\mathbf{I}=\left(  I_{1},I_{2},\ldots,I_{n}\right)  $.\ As mentioned above, a
tontine fund for this group is characterized by $\left(  \mathbf{I}%
,\boldsymbol{\pi},\mathbf{f}\right)  $.\ Let us now assume that the $n$
participants and the administrator agree on a probability vector
\[
\mathbf{p}=\left(  p_{1},\text{ }p_{2},\ldots,\text{ }p_{n}\right)  ,
\]
where $p_{i},i=1,2,\ldots,n$ is the survival probability of participant $i$,
i.e. $p_{i}=P\left[  I_{i}=1\right]  $.\ The vector $\mathbf{p}$ of the
survival probabilities has to be interpreted as an 'agreed vector', which may
be different from the 'real vector' of the survival probabilities of the participants.\ 

We assume that the share allocation vector $\mathbf{f}$ is a function of the
contribution vector $\boldsymbol{\pi}$\ and the probability vector
$\mathbf{p}$:
\[
\mathbf{f=f}\left(  \boldsymbol{\pi},\mathbf{p}\right)  =\left(  f_{1}\left(
\boldsymbol{\pi},\mathbf{p}\right)  ,f_{2}\left(  \boldsymbol{\pi}%
,\mathbf{p}\right)  ,\ldots,f_{n}\left(  \boldsymbol{\pi},\mathbf{p}\right)
\right)  ,
\]
where the value of $f_{i}\left(  \boldsymbol{\pi},\mathbf{p}\right)  $
corresponds to the number of tontine shares received by person $i$ in the
tontine fund $\left(  \mathbf{I},\boldsymbol{\pi},\mathbf{f}\right)  $.\ 

Each $f_{i}\left(  \boldsymbol{\pi},\mathbf{p}\right)  $, $i=1,2,\ldots,n$,
can be interpreted as a measure of the 'risk exposure' of the corresponding
participant, taking into account the information on initial investments and
survival probabilities of all participants.\ We call the function $\mathbf{f}$
an \textit{internal share allocation scheme} in the sense that the allocated
number of shares only depends on internal information of the pool, i.e. on
$\left(  \boldsymbol{\pi},\mathbf{p}\right)  $.\ More generally, one could
also consider more complex share allocation schemes, where the number of
allocated shares does not only depend on $\boldsymbol{\pi}$ and $\mathbf{p}$,
but also on other deterministic information and time $1$ observable random
variables, such as the state of the economy, the occurrence (or not) of a
pandemic over the coming year, the precise magnitude of medical inflation over
the coming year, etc.\ 

From (\ref{MD5}), we find that the payouts $W_{i}$ can be expressed as
follows:
\begin{equation}
W_{i}=\left(  1+R\right)  \times\left(  \sum_{j=1}^{n+1}\pi_{j}\right)
\times\frac{f_{i}\left(  \boldsymbol{\pi},\mathbf{p}\right)  }{\sum
_{j=1}^{n+1}f_{j}\left(  \boldsymbol{\pi},\mathbf{p}\right)  \times I_{j}%
}\times I_{i},\qquad\text{for }i=1,2,\ldots,n+1. \label{D25}%
\end{equation}

At the set-up of the tontine fund with an internal share allocation rule, an
assumption about (or agreement on) the survival probabilities of the
participants are required to be able to define the payouts. But once the
set-up is ready, i.e. once the tontine fund is launched, it only needs an
administrator who collects info about the survival or death of the
participants, from which he can then determine each payout $W_{i}$ via formula
(\ref{D25}).\ Notice that any strictly positive choice for $f_{n+1}\left(
\boldsymbol{\pi},\mathbf{p}\right)  $\ can be made as the payouts are not
dependent on this choice.

In certain situations, it may be reasonable to impose a linear behaviour
between the number of allocated tontine shares $f_{i}$\ and the initial
investment $\pi_{i}$, when the survival probability $p_{i}$ is fixed.\ That
means that it may be appropriate to assume that
\begin{equation}
f_{i}\left(  \boldsymbol{\pi},\mathbf{p}\right)  =\pi_{i}\times g\left(
p_{i}\right)  ,\qquad\text{for }i=1,2,\ldots,n\text{.} \label{D8}%
\end{equation}
for all participants $i$,\ where $g$ is strictly positive and decreasing (or
increasing, or something else) in the survival probability $p_{i}$. A
decreasing $g$ corresponds to the mathematical translation that 'a participant
with a smaller survival probability receives a larger number of tontine shares
than a person with the same initial investment but higher survival
probability.'\ Such an approach is inspired by the idea that the person with a
smaller survival probability has a higher chance of losing his initial
investment.\ On the other side, in case one imposes an increasing $g$ that
means that the allocation rule is such that it favours participants with
higher survival probabilities. This might be a desirable property in a closely
connected social group, and would be in the hands of the scheme's architects.
Finally, notice that we assume here that $g$ is not participant-specific and
hence, does not depend on $i$. More generally, one could introduce a
participant-specific function $g_{i}$.

In case (\ref{D8}) holds, we have that (\ref{D25}) transforms into
\begin{equation}
W_{i}=\left(  1+R\right)  \times\left(  \sum_{j=1}^{n+1}\pi_{j}\right)
\times\frac{\pi_{i}\times g\left(  p_{i}\right)  }{\sum_{j=1}^{n+1}\pi
_{j}\times g\left(  p_{j}\right)  \times I_{j}}\times I_{i},\qquad\text{for
}i=1,2,\ldots,n+1. \label{D9}%
\end{equation}

Hereafter, we introduce some important special cases of the share allocation
scheme defined in (\ref{D25}).

\begin{example}
The \textbf{DM} allocation scheme. \newline Let us assume the internal share
allocation scheme (\ref{D8}), where $g\left(  p_{i}\right)  =1/p_{i}$.\ In
other words, we consider the following tontine share allocation scheme:
\begin{equation}
f_{i}^{\text{DM}}\left(  \boldsymbol{\pi},\mathbf{p}\right)  =\frac{\pi_{i}%
}{p_{i}},\qquad i=1,2,\ldots,n+1\text{.} \label{MD2}%
\end{equation}
In this case, the payouts (\ref{D9}) of the participants are given by%
\begin{equation}
W_{i}^{\text{DM}}=\left(  1+R\right)  \times\left(  \sum_{j=1}^{n+1}\pi
_{j}\right)  \times\frac{\frac{\pi_{i}}{p_{i}}}{\sum_{j=1}^{n+1}\frac{\pi_{j}%
}{p_{j}}\times I_{j}}\times I_{i},\qquad\text{for }i=1,2,\ldots,n+1,
\label{MD3}%
\end{equation}
which corresponds with the tontine fund payouts $W_{i}$\ that we introduced in
(\ref{T3'}).\ A special case was considered in (\ref{rule}).\ Notice that we
have chosen $f_{n+1}^{\text{DM}}\left(  \boldsymbol{\pi},\mathbf{p}\right)
=\frac{\pi_{n+1}}{p_{n+1}}$, with $p_{n+1}=P\left[  I_{n+1}=1\right]  $, but
any other strictly positive value of $f_{n+1}^{\text{DM}}\left(
\boldsymbol{\pi},\mathbf{p}\right)  $ will lead to the same payouts.\ \newline
A motivation for this allocation (\ref{MD2}) in terms of traditional insurance
benefits was given in the previous section.\ 

It's interesting to note that a rather special case of the above will arise if
all participants are required to have the same risk exposure in the sense that
$\frac{\pi_{i}}{p_{i}}=c$, where $c$, is a given constant. In that special
case, for a given (or assumed) vector of survival probabilities, the
investments are then given by:
\[
\pi_{i}=c\times p_{i},\qquad i=1,2,\ldots,n+1,
\]
and (\ref{MD3}) reduces to:
\[
W_{i}=\left(  1+R\right)  \times\left(  \sum_{j=1}^{n+1}\pi_{j}\right)
\times\frac{I_{i}}{\sum_{j=1}^{n+1}I_{j}},\qquad\text{for }i=1,2,\ldots,n+1
\]

\end{example}

\begin{example}
The \textbf{T} allocation scheme. \newline Consider the internal share
allocation scheme (\ref{D8}) with $g(p)\equiv1$:\
\[
f_{i}^{\text{T}}\left(  \boldsymbol{\pi},\mathbf{p}\right)  =\pi_{i}%
,\qquad\text{for }i=1,2,\ldots,n+1\text{.}%
\]
Then (\ref{D9}) becomes%
\begin{equation}
W_{i}^{\text{T}}=\left(  1+R\right)  \times\left(  \sum_{j=1}^{n+1}\pi
_{j}\right)  \times\frac{\pi_{i}}{\sum_{j=1}^{n+1}\pi_{j}\times I_{j}}\times
I_{i},\qquad\text{for }i=1,2,\ldots,n+1. \label{D21}%
\end{equation}
Notice that from (\ref{D7}), it follows that the time -$\ 0$ value of a
tontine share, notation $S^{\text{T}}\left(  0\right)  $ is now given by
\[
S^{\text{T}}(0)=\frac{\sum_{j=1}^{n+1}\pi_{j}}{\sum_{j=1}^{n}\pi_{j}}%
=1+\frac{\pi_{n+1}}{\sum_{j=1}^{n}\pi_{j}}.
\]
From (\ref{D21'}) it follows then that (\ref{D21})\ can be rewritten as
follows:%
\begin{align}
W_{i}^{\text{T}}  &  =\left(  1+\frac{\pi_{n+1}}{\sum_{j=1}^{n}\pi_{j}%
}\right)  \times\left(  1+R\right)  \times\pi_{i}\times\left(  1+\frac
{\sum_{j=1}^{n+1}\pi_{j}\times\left(  1-I_{j}\right)  -\pi_{n+1}}{\sum
_{j=1}^{n+1}\pi_{j}\times I_{j}}\right)  \times I_{i},\nonumber\\
\text{for }i  &  =1,2,\ldots,n+1. \label{D21a}%
\end{align}
\newline When $\pi_{n+1}=0$, formula (\ref{D21a}) remains to hold, provided we
replace $\pi_{n+1}$ in $\frac{\sum_{j=1}^{n+1}\pi_{j}\times\left(
1-I_{j}\right)  -\pi_{n+1}}{\sum_{j=1}^{n+1}\pi_{j}\times I_{j}}$ by a
strictly positive value $f_{n+1}$. In this case, we find that%
\[
W_{i}^{\text{T}}=\left(  1+R\right)  \times\pi_{i}\times\left(  1+\frac
{\sum_{j=1}^{n}\pi_{j}\times\left(  1-I_{j}\right)  -f_{n+1}\times I_{n+1}%
}{\sum_{j=1}^{n}\pi_{j}\times I_{j}+f_{n+1}\times I_{n+1}}\right)  \times
I_{i},\qquad\text{for }i=1,2,\ldots,n+1.
\]
In case at least one participant survives, i.e. $I_{n+1}=0$, we have that
\[
\left(  W_{i}^{\text{T}}\mid I_{n+1}=0\right)  =\left(  1+R\right)  \times
\pi_{i}\times\left(  1+\frac{\sum_{j=1}^{n}\pi_{j}\times\left(  1-I_{j}%
\right)  }{\sum_{j=1}^{n}\pi_{j}\times I_{j}}\right)  \times I_{i}%
,\qquad\text{for }i=1,2,\ldots,n.
\]

\end{example}

In the questionnaire survey that we noted in the early part of the paper, one
of the replies\footnote{Private communication from Bertrand Tavin.} that we
received was the above-noted formula, and which we denote as Tavin allocation
scheme. That scheme favours younger participants, individuals with higher
survival probabilities. Indeed, consider two participants, $i$ and $j$, who
invest the same amount $\pi_{i}=$ $\pi_{j}$, but the first one is younger than
the second one in the sense that $p_{i}>$ $p_{j}$.\ Then obviously, the
younger person is favoured as in the case of survival, both receive the same
amount, although the younger one has a higher survival probability. To
paraphrase Tavin (2023):

\begin{quote}
``...In this allocation, the recorded amount upon survival is only driven by
the agent's initial stake compared to the others' stakes. This system plays a
role in terms of the welfare of the social group.\ Namely, there is a
reallocation of wealth (the total amount in the fund) that is favourable to
those who are likely to survive long after the liquidation of the tontine,
compared to the risk-return-based allocation, which favours the agents who are
likely not to survive long after the liquidation of the tontine.\ This system
increases the group's welfare if we look at the welfare obtained by the
surviving agents after time $1$. Conditional on survival at time $1$, those
likely to live long after time $1$ need more (because they probably need to
take care of an elder parent or children) and are more likely to have projects
that benefit the social group (e.g.\ opening or financing a business).\ On the
other hand, the agent who is not likely to survive long after time $1$ will
probably not have enough time to enjoy the received amount..."
\end{quote}

\begin{example}
Consider the share allocation rule with:%
\[
f_{i}\left(  \boldsymbol{\pi},\mathbf{p}\right)  =\frac{1}{p_{i}},\qquad
i=1,2,\ldots,n\text{.}%
\]
In this case, we find from (\ref{D25}) that
\begin{equation}
W_{i}=\left(  1+R\right)  \times\left(  \sum_{j=1}^{n+1}\pi_{j}\right)
\times\frac{\frac{1}{p_{i}}}{\sum_{j=1}^{n+1}\frac{I_{j}}{p_{j}}}\times
\frac{I_{i}}{p_{i}},\qquad\text{for }i=1,2,\ldots,n+1\text{.} \label{D30}%
\end{equation}
This rule favors 'poorer' participants (i.e. participants who invest less).
Indeed, consider two persons $i$ and $j$ with initial investments $\pi_{i}<$
$\pi_{j}$. Suppose that both have the same survival probability.\ Then in case
of survival both receive the same amount, whereas person $i$ invested less.
\end{example}

\begin{example}
The \textbf{DR} allocation scheme, following the work of Denuit \& Robert
(2023). \newline Consider the uniform rule with%
\[
f_{i}^{\text{DR}}=1,\qquad i=1,2,\ldots,n+1.
\]
In this case, we find from (\ref{D25}) that
\begin{equation}
W_{i}^{\text{DR}}=\left(  1+R\right)  \times\left(  \sum_{j=1}^{n+1}\pi
_{j}\right)  \times\frac{I_{i}}{\sum_{j=1}^{n+1}I_{j}},\qquad\text{for
}i=1,2,\ldots,n+1\text{.} \label{T7}%
\end{equation}
From (\ref{D7}), we find that $S^{\text{DR}}(0)=\frac{\sum_{j=1}^{n+1}\pi_{j}%
}{n}$.\ Hence, from (\ref{D21'}), we find that
\[
W_{i}^{\text{DR}}=\left(  1+R\right)  \times\frac{\sum_{j=1}^{n+1}\pi_{j}}%
{n}\times\left(  1+\frac{\sum_{j=1}^{n+1}\left(  1-I_{j}\right)  -1}%
{\sum_{j=1}^{n+1}I_{j}}\right)  \times I_{i},\qquad\text{for }i=1,2,\ldots
,n+1\text{.}%
\]

This scheme is advantageous for individuals who are younger and poorer.
Suppose there are two people, $i$ and $j$. If $i$ is younger and has a higher
chance of survival, but is poorer and pays less to the tontine fund than $j$,
but pays less, he will receive the same payout money if they both survive.
Therefore, the tontine arrangement benefits the younger and less affluent
person $i$. A similar scheme can be found in Denuit \& Robert (2023), with the
difference being that they make their rule fair by returning the initial
investments if all participants pass away and defining the initial investments
so that the allocations are actuarially fair.
\end{example}

While this isn't the main focus of our paper, we should note the following
about the (rather loaded term) \textquotedblleft actuarial
fairness\textquotedblright. Namely, taking into account Theorem 4, the above
discussed allocation schemes or arrangements -- although not generally fair to
any given individual -- can be made collectively fair (in the sense of
Definition 3), or perhaps the proper word is actuarially \textquotedblleft
just\textquotedblright\ to add another term of the growing lexicon, by
introducing the administrator.

\section{Summary and Conclusion}

The launching pad for this paper -- both conceptually and in practice --
revolves around the many justifiable ways in which a group of heterogenous
individuals could in theory, share the proceeds of a (longevity) risk-pooling
agreement. We motivated the examples using the concept of a one-period
tontine, a product that is enjoying a resurgence of interest worldwide, both
in academia and industry. If members of this heterogeneous group invested
unequal amounts into the tontine pool, our (small pool) numerical example made
the multiplicity of possible solutions evident. Therefore, the primary
contribution of this paper is to argue that the payout structure for a tontine
fund can be quite comprehensive, catering to a broad range of groups wishing
to share longevity risks without the interference of an external entity
assuming the risk of insolvency.

These insights are particularly beneficial for closely-knit smaller groups
aiming to redistribute wealth from their older, wealthier members to their
younger, less prosperous counterparts. In such scenarios, the emphasis isn't
on actuarial fairness or the magnitude of the administrator's contribution.
Instead, the focus is on the collective benefit of the group, as the
administrator embodies the group's interests, and their contributions directly
benefit the group.

Our methodology can accommodate larger groups, even when participants may not
share social connections or interpersonal ties. In these situations surviving
members ought to be compensated based on the actuarial risks they've accepted
and been exposed to. Individuals with a lower likelihood of survival should be
entitled to a more substantial reward. An external entity, like an insurance
company or a government regulator, administrator in these contexts. They could
also contribute in scenarios with a significant likelihood of none of the
participants surviving.

In sum, while the objective of modern tontines, and more generally uninsured
decumulation products (UDP),\footnote{This is the term recently introduced by
Canadian regulators to describe the arrangements of this sort. See:
\emph{https://www.fsrao.ca/regulation/guidance/understanding-decumulation-products}%
.} is to eliminate the costly capital associated with insurance
\textbf{guarantees}, we are not advocating the elimination of insurance
\textbf{regulators.} Rather, under these arrangements, the role of the
regulator would be to administer the fund -- in exchange for \textquotedblleft
a piece of the action\textquotedblright\ -- which would serve two distinct
roles. First, oversight. They would ensure all participants in the scheme were
abiding by their obligations and commitments. Second, and just as importantly,
administrators in the scheme would make it collectively actuarially fair, that
is, socially just.

The next step is to extend this one-period framework to a multiperiod tontine
fund, which would be constructed as a sequence of linked one-period funds.
Defining the relations between indicator vectors, premium vectors, share
allocation vectors, and the all-important payout vectors in consecutive
periods is left in an honoured tradition for future research. In the same
category of plans for future research, we leave the discussion of allowing
$\pi_{i}$ and even $f_{i}$ to equal zero, allowing certain groups to avoid
paying (and still benefiting) or not benefiting (and still paying.) Examples
would be targetted demographic groups such as the young and old respectively.

We conclude by noting that the single-period tontine fund, which is described
within the body of this paper can be treated or viewed as a special case of
(what we will call) \textit{compensation-based} decentralized risk-sharing
(DRS) arrangements, where at time $t=0$ one contributes (deterministic)
premiums (or investments) and at time $t=1$ one receives (random)
compensations, which are set such that the risk-sharing scheme is
self-financing. Like the literature on tontines, there is a growing literature
on this type of DRS. A somewhat less obvious insight is that the single period
tontine fund can also be expressed in the form of (what we will call) a
\textit{contribution-based} DRS scheme, which is characterized by time $1$
(random) contributions and time $1$ (random) benefits or claims, and where the
compensations are determined such that the risk-sharing scheme is again
self-financing. There is also a growing literature on this type of DRS. For an
overview of a unified theory of DRS, we refer to Feng (2023) and the
references in that book.\ The above-mentioned observations are further
explored in some detail in the Appendix.

\subsection{Acknowledgement \& Thanks.}

The authors would like to acknowledge many helpful comments and input, as well
as the efforts involved in responding to our survey questionnaire, from Robert
Bertrand, Servaes van Bilsen, Andrew Carrothers, Doug Chandler, Michel Denuit,
Runhuan Feng, Faisal Habib, Peter Hieber, Aleksi Leeuwenkamp, Andrew
McDiarmid, Kent McKeever, Branislav Nikolic, Christian Robert and Thomas
Salisbury. The early drafts of this paper benefited (immensely) from their
insights and individuals have been noted, and thanked in specific places, as
deemed appropriate. Finally, Jan Dhaene gratefully acknowledges funding from
FWO and F.R.S.-FNRS under the Excellence of Science (EOS) programme, project
ASTeRISK (40007517), and Moshe A. Milevsky acknowledges funding from the IFID
Centre (Grant \# 2023.12).

\newpage

\section*{References}

\begin{enumerate}
\item Ayuso, M., Bravo, J. M., \& Holzmann, R. (2017). Addressing Longevity
Heterogeneity in Pension Scheme Design. \textit{Journal of Finance and
Economics}, 6(1), 1-21.


\item Bernard, Feliciangeli \& Vanduffel (2022), Optimal Risk Sharing in an
Actuarially Unfair Tontine, Presentation at 8th Workshop on Risk Management
and Insurance Research (RISK2022), Working paper.

\item Bernhardt, T., \& Donnelly, C. (2019). Modern tontine with bequest:
Innovation in pooled annuity products. \textit{Insurance: Mathematics and
Economics}, 86, 168-188.

\item Bernhardt, T. \& Qu, G. (2023) Wealth heterogeneity in a closed pooled
annuity fund, \textit{Scandinavian Actuarial Journal},

\item Bravo J., M. Ayuso, R. Holzmann, \& E. Palmer (2023), Intergenerational
actuarial fairness when longevity increases, \textit{Insurance: Mathematics
and Economics}, in press.

\item Chetty, R., Stepner, M., Abraham, S., Lin, S., Scuderi, B., Turner, N.,
Bergeron, A., \& Cutler, D. (2016). The Association Between Income and Life
Expectancy in the United States, 2001-2014. \textit{JAMA}, 315(16), 1750-1766.

\item Chen, A., Chen, Y., \& Xu, X. (2022). Care-dependent Tontines.
\textit{Insurance: Mathematics and Economics.}

\item Cheung, K. C., \& Lo, A. (2014). Characterizing mutual exclusivity as
the strongest negative multivariate dependence structure. \textit{Insurance:
Mathematics and Economics}, 55, 180-190.

\item Coppola, M., Russolillo, M., \& Simone, R. (2022). On the evolution of
the gender gap in life expectancy at normal retirement age for OECD countries.
\textit{Genus}, 78(1).

\item Couillard, B. K., Foote, C. L., Gandhi, K., Meara, E., \& Skinner, J.
(2021). Rising Geographic Disparities in US Mortality. \textit{J Econ
Perspect}, 35(4), 123-146.

\item Denuit, M., Robert, C.Y. (2023).\ Endowment contingency funds for mutual
aid and public financing, \emph{working paper}.

\item Denuit, M., Hieber, P., \& Robert, C.Y. (2022). Mortality credits within
large survivor funds. \textit{ASTIN Bulletin}, 52(3), 813-834.

\item Denuit M., Dhaene J. (2012). Convex order and comonotonic conditional
mean risk sharing. \textit{Insurance: Mathematics and Economics} 51, 249-256.

\item Denuit M., Dhaene J, \& Robert C.Y. (2022). Risk-sharing rules and their
properties, with applications to peer-to-peer insurance. \textit{Journal of
Risk and Insurance}, 89(3), 615-667.

\item Denuit M., Dhaene J, Ghossoub M., \& Robert C.Y. (2022). Comonotonicity
and Pareto optimality, with application to collaborative
insurance.\ \textit{LIDAM Paper 2023/05}.

\item Dhaene, J., Denuit, M.\ (1999).\ The safest dependency structure among
risks.\ Insurance: Mathematics \& Economics, 25, 11-21.

\item Donnelly, C. (2018). Methods of Pooling Longevity Risk, \emph{Actuarial
Research Centre.}

\item Donnelly, C., Guillien, M., \& Nielsen, J. P. (2014). Bringing cost
transparency to the life annuity market. \textit{Insurance: Mathematics and
Economics}, 56, 14-27.

\item Dudel, C., \& van Raalte, A. A. (2023). Educational inequalities in life
expectancy: measures, mapping, meaning. \textit{J Epidemiol Community Health},
77(7), 417-418.

\item Feng, R. and P. Liu (2024), A Unified Theory of Multi-Period
Decentralized Insurance and Annuities, \emph{working paper}.

\item Feng, R. (2023). \emph{Decentralized Insurance: Technical Foundation of
Business Models}, Springer International Publishing, https://doi.org/10.1007/978-3-031-29559-1

\item Finegood, E. D., Briley, D. A., Turiano, N. A., Freedman, A., South, S.
C., Krueger, R. F., Chen, E., Mroczek, D. K., \& Miller, G. E. (2021).
Association of Wealth With Longevity in US Adults at Midlife. \textit{JAMA
Health Forum}, 2(7), e211652.

\item Forman, J. B., \& Sabin, J. M. (2015). Tontine pensions.
\textit{University of Pennsylvania Law Review}, 163(3), 755-831.

\item Himmelstein, K. E. W., Lawrence, J. A., Jahn, J. L., Ceasar, J. N.,
Morse, M., Bassett, M. T., Wispelwey, B. P., Darity, W. A., Jr., \&
Venkataramani, A. S. (2022). Association Between Racial Wealth Inequities and
Racial Disparities in Longevity Among US Adults and Role of Reparations
Payments, 1992 to 2018. \textit{JAMA Netw Open}, 5(11), e2240519.

\item Jiao, Z., Kou, S., Liu, Y., \& Wang, R. (2022). An axiomatic theory for
anonymized risk sharing. \textit{arXiv:2208.07533.}

\item Kinge, J. M., Modalsli, J. H., Overland, S., Gjessing, H. K., Tollanes,
M. C., Knudsen, A. K., Skirbekk, V., Strand, B. H., Haberg, S. E., \& Vollset,
S. E. (2019). Association of Household Income With Life Expectancy and
Cause-Specific Mortality in Norway, 2005-2015. \textit{JAMA}, 321(19), 1916-1925.

\item Lauzier, J. G., Lin, L., \& Wang, R. (2024). Negatively dependent
optimal risk sharing. \textit{arXiv preprint arXiv:2401.03328.}

\item Li, H., \& Hyndman, R. J. (2021). Assessing mortality inequality in the
U.S.: What can be said about the future? \textit{Insurance: Mathematics and
Economics}, 99, 152-162.

\item Lin, T. Y., Chen, C. Y., Tsao, C. Y., \& Hsu, K. H. (2017). The
association between personal income and aging: A population-based 13-year
longitudinal study. \textit{Arch Gerontol Geriatr}, 70, 76-83.

\item Mackenbach, J. P., Valverde, J. R., Bopp, M., Bronnum-Hansen, H.,
Deboosere, P., Kalediene, R., Kovacs, K., Leinsalu, M., Martikainen, P.,
Menvielle, G., Regidor, E., \& Nusselder, W. (2019). Determinants of
inequalities in life expectancy: International comparative study of 8 risk
factors. \textit{Lancet Public Health}, 4(10), 529-537.

\item McKeever, K. (2009). A short history of tontines. \textit{Fordham J.
Corp. \& Fin. L.}, 15, 491.

\item Milevsky, M. A., \& Salisbury, T. S. (2015). Optimal retirement income
tontines. \textit{Insurance: Mathematics and Economics}, 64, 91-105.

\item Milligan, K., \& Schirle, T. (2021). The evolution of longevity:
Evidence from Canada. \textit{Canadian Journal of Economics}, 54(1), 164-192.

\item Perez-Salamero Gonzalez, J. M., Regulez-Castillo, M., Ventura-Marco, M.,
\& Vidal-Melia, C. (2022). Mortality and life expectancy trends in Spain by
pension income level for male pensioners in the general regime retiring at the
statutory age, 2005-2018. \textit{Int J Equity Health}, 21(1), 96.

\item Piggott, J., Valdez, E. A., \& Detzel, B. (2005). The simple analytics
of a pooled annuity fund. \textit{Journal of Risk and Insurance}, 72(3), 497-520.

\item Pitacco, E. (2019). Heterogeneity in mortality: a survey with an
actuarial focus. \textit{European Actuarial Journal}, 9(1), 3-30.

\item Sanzenbacher, G. T., Webb, A., Cosgrove, C. M., \& Orlova, N. (2019).
Rising Inequality in Life Expectancy by Socioeconomic Status. \textit{North
American Actuarial Journal}, 25(sup1), S566-S581.

\item Shi, J., \& Kolk, M. (2022). How Does Mortality Contribute to Lifetime
Pension Inequality? Evidence From Five Decades of Swedish Taxation Data.
\textit{Demography}, 59(5), 1843-1871.

\item Simonovits, A., \& Lack, M. (2023). A simple estimation of the longevity
gap and redistribution in the pension system. \textit{Acta Oeconomica}, 73(2), 275-284.

\item Sloan, F. A., Ayyagari, P., Salm, M., \& Grossman, D. (2010). The
longevity gap between Black and White men in the United States at the
beginning and end of the 20th century. \textit{Am J Public Health}, 100(2), 357-363.

\item Strozza, C., Vigezzi, S., Callaway, J., Kashnitsky, I., Aleksandrovs,
A., \& Vaupel, J. W. (2022). Socioeconomic inequalities in survival to
retirement age: a register-based analysis.

\item Stamos, M. Z. (2008). Optimal consumption and portfolio choice for
pooled annuity funds. \textit{Insurance: Mathematics and Economics}, 43(1), 56-68.

\item Weinert, J. H., \& Grundel, H. (2021). The modern tontine: An innovative
instrument for longevity risk management in an aging society. \textit{European
Actuarial Journal}, 11(1), 49-86.
\end{enumerate}

\newpage

\section{Appendix: Tontines \& Decentralized Risk Sharing}

Two areas of insurance research whose literature has grown lately are: (i)
decentralized risk-sharing systems and (ii) retirement tontine arrangements.
The former system is characterized by a risk-sharing rule without solvency
capital or a formal third-party guarantor. The latter investment arrangements
are designed as alternatives to life annuities, characterized by a group
covenant in which longevity risk is pooled and mortality credits aren't
guaranteed. This appendix models the concepts of single period decentralized
risk-sharing and tontine endowments -- which was the focus on the main body of
the paper -- in one unified framework and \textquotedblleft
proves\textquotedblright\ they can be viewed as mathematical duals of each
other. While others, and especially the recent work by Feng \& Liu (2024),
have hinted at these connections, our objective in this appendix is to
continue the work of unifying two disparate literatures under one banner.

\subsection{A Brief Review}

To illustrate and explain the essence of \emph{decentralized risk sharing}
(DRS) in which nothing is guaranteed and therefore, no solvency risk capital
is required, we begin by reviewing the basics of classical insurance theory
and thus set notation and terminology as well. To that end, consider a pool of
$n>1$ economic agents or individual policyholder participants. Each one of
these $n$ policyholders purchases an insurance contract at time $t=0$, which
entitles him or her to a random claim amount denoted by $X_{i}$, measurable
and payable at time $T=1$. This claim can be a random loss related to a
well-defined peril (e.g. hospitalization-related expenses) , or it can be a
random benefit contingent on the occurrence of a well-defined event (e.g. a
pure endowment, which entitles the survivor to a predefined amount of money)
due at time $1$.

For simplicity, we assume that the above-noted time interval $\left[
0,1\right]  $, is a calendar year, but one could obviously adopt and modify
what follows to larger time intervals of $T=5$, $T=10$, or $T=20$ years, etc.
Rather, the point here is that our setup is a one-period model with no
intermediate cash-flows or payouts.

Suppose further that each policyholder $i$ pays the insurer (to be explained)
a premium amount $\pi_{i}$ at time $t=0$, to acquire said protection or
benefit and the totality of all premiums $\sum_{i=1}^{n}+\pi_{{i}}$ is
collectively invested in a risk-free (i.e. default-free) asset, subject to a
return of $R$ per period. In other words $\$1$ grows to $\$\left(  1+R\right)
^{r}$ at the end of the period.

Now let's focus on the premiums $\pi_{i}$ and how they are determined. The job
of an actuary -- regardless of the particular risk being insured -- is to
ensure that the probability that the aggregate claims $\sum_{i=1}^{n}X_{i}$
due at time $1$ are not larger than the available assets at time $1$ for the
insurance portfolio under consideration.\ Here, the available assets consist
of the time-$1$ value of the premiums $\left(  1+R\right)  \sum_{i=1}^{n}%
\pi_{i}$ and the time-$1$ value $(1+R)SC$ of the solvency capital $SC$ to be
set up at time $0$.\ In particular, the insurer will become insolvent and end
the period in bankruptcy if the following undesirable event occurs:%
\[
\left(  1+R\right)  \left(  \sum_{i=1}^{n}\pi_{i}+SC\right)  <\sum_{i=1}%
^{n}X_{i}%
\]
While the above condition is rather obvious and intuitive, namely that
premiums and solvency capital haven't accumulated enough to pay all the
claims, the opposite condition can be re-written and expressed in the
following (equally obvious) manner:%
\begin{equation}
\text{\textbf{ Insurance Solvency Condition: \hspace{0.2in}} }\sum_{j=1}%
^{n}X_{j}\leq\left(  1+R\right)  \left(  \sum_{j=1}^{n}\pi_{j}+SC\right)  .
\label{E1a}%
\end{equation}

This intuition is straightforward namely the insurer (i.e. guarantor) can
fulfill his/her liabilities \emph{if-and-only-if} the total claims to be paid
at time $T=1$ are \textbf{not} larger than the accumulated value of the
premiums and solvency capital, that is, the assets available to pay all claims.

Classical insurance is of the form of \textquotedblleft
centralized\textquotedblright\ (versus decentralized) risk sharing. This
implies or is associated with a risk-sharing mechanism under which individual
losses faced by policyholders of the pool are transferred to a central
insurer. Every single one of the $n$ policyholders is compensated
\emph{ex-post} from the insurer for the experienced loss, which we denoted by
$X_{i}$. In return for that total and absolute coverage, the insurer charges
an insurance premium \emph{ex-ante}, paid by each of the $n$ insured
policyholders. Now, the premiums $\pi_{i}$ themselves will follow from an
appropriate \emph{premium principle}, selected in such a manner so that the
probability of the event that the sum of all accumulated premiums and solvency
capital exceeds the aggregate loss of the insurance portfolio is sufficiently
high (e.g.\ 99.5\%). Now, and this is key, the centralized approach with
\emph{ex-ante} premiums requires capital to be set up by the insurer to be
able to meet his \emph{ex-post} obligations. Premiums should be large enough
so that the solvency capital does not have to be extremely high (which the
owners of the insurance company will not want), but premiums should not be too
large in order to remain competitive in the market.

\subsection{DRS via Compensations}

Consider a pool of $n+1$ individual random future claims $X_{i}$, each related
to to a well-defined peril or contingent benefit. The vector or pool of all
risks (benefits or losses) is denoted by $\mathbf{X}$:
\[
\mathbf{X}=\left(  X_{1},X_{2},\ldots,X_{n+1}\right)  .
\]
And, we assume that at time $0$, each participants pays a premium (or invests
an amount) of size $\pi_{i}$, with the vector of all premiums labeled the
\underline{\textbf{premium vector}}:
\[
\boldsymbol{\pi}=\left(  \pi_{1},\pi_{2},\ldots,\pi_{n+1}\right)  .
\]
Important to mention is that participant $n+1$ is the \textit{administrator}%
.\ He is the person responsible for the management of the fund (collecting
investments, paying compensations).\ Important is that the administrator can
also take part in the DRS scheme, by paying a premium $\pi_{n+1}$ and
receiving an appropriate compensation, see further.\ 

Decentralized Risk-Sharing (henceforth, DRS) refers to risk-sharing mechanisms
under which the participants in the pool share or allot the risks among each
other in such a way that the administrator of the payments incurs no
insolvency risk. In fact, to avoid any confusion with centralized
risk-sharing, which is only missing a \emph{de}, we introduce the concept of a
\emph{community}, which administers the scheme but doesn't guarantee anything.
One can think of this \emph{community} as a collection of individuals. And,
even if they are properly regulated, have formal administrators, and might
even be incorporated, they differ from insurance \emph{companies.}

To achieve that objective and figure out the allotment, premium payments are
made at time $T=0$, but these aren't premiums in the classical insurance
sense. Rather, each participant $i$ from the pool of size $n+1$ is eventually
and only partially compensated \emph{ex-post} for their loss $X_{i}$. In other
words, they do not receive of get reimbursed for the total amount $X_{i}$. The
pool -- or perhaps better labeled the community -- will pay or compensate in
the amount $W_{i}\left(  \mathbf{X}\right)  $, for a claim $X_{i}$. Now, the
vector of all compensations is called the \underline{\textbf{Compensation
Vector}}:%

\[
\mathbf{W}\left(  \mathbf{X}\right)  =\left(  W_{1}\left(  \mathbf{X}\right)
,W_{2}\left(  \mathbf{X}\right)  ,\ldots,W_{n+1}\left(  \mathbf{X}\right)
\right)
\]

The function which transforms any $\mathbf{X}$ into the compensation vector
$\mathbf{W}\left(  \mathbf{X}\right)  $ is called a compensation-based
risk-sharing rule. We do not necessarily assume that $\mathbf{W}$ is a
function from $\mathbb{R}^{n+1}$ to $\mathbb{R}^{n+1}$. Rather, we only assume
that $\mathbf{W}$ is a function from a given set of $\left(  n+1\right)
$-dimensional random vectors defined on a given probability space to this same
set of random vectors. Hereafter, we will denote the compensation-based DRS
described above by $\left(  \boldsymbol{\pi},\mathbf{W}\left(  \mathbf{X}%
\right)  \right)  $.

In order to avoid insolvency risk, the risk-sharing rule is such that the
following compensation-based condition is fulfilled:
\begin{equation}
\text{\textbf{Compensation-Based Solvency Condition}: }\sum_{j=1}^{n+1}%
W_{j}\left(  \mathbf{X}\right)  =\left(  1+R\right)  \sum_{j=1}^{n+1}\pi_{j},
\label{E9}%
\end{equation}
which means that the sum of all compensations is exactly equal to the time
$t=1$ value of the sum of all premiums paid by the participants and the
administrator. It should be clear that in the compensation-based risk-sharing
scheme $\left(  \boldsymbol{\pi},\mathbf{W}\left(  \mathbf{X}\right)  \right)
$, the time $1$ value of the total cash-inflow for participant $i$ is given
by:
\begin{equation}
W_{i}\left(  \mathbf{X}\right)  -\left(  1+R\right)  \times\pi_{i}.
\label{E10}%
\end{equation}

As an example, consider the case where for each participant $i$ with claim
$X_{i}$, the compensation function $W_{i}\left(  \mathbf{X}\right)  $ follows
from
\begin{equation}
W_{i}^{\text{prop}}\left(  \mathbf{X}\right)  =\left(  1+R\right)
\times\left(  \sum_{j=1}^{n+1}\pi_{j}\right)  \times\frac{X_{i}}{\sum
_{j=1}^{n+1}X_{j}}\qquad\text{for }i=1,2,\ldots,n+1. \label{E11}%
\end{equation}
This compensation-based risk-sharing rule is called the
\underline{\textbf{proportional risk-sharing rule}}: Given the aggregate
claims $\sum_{j=1}^{n+1}X_{j}$, each participant $i$ receives a compensation
proportional to his observed claim $X_{i}$, where the proportional factors are
determined such that the full compensation condition (\ref{E9}) is satisfied.
The key insight in the tontine literature is that similar to the Markowitz
trade-off between investment risk and return, policyholders should be able to
choose between insurance with (costly) guarantees versus pooling arrangements,
or perhaps even a mixture.

Now, just as we encountered for the one-period tontine fund described within
the body of the paper, in general there is non-zero probability that
$\sum_{j=1}^{n}X_{j}$ is equal to $0$ in (\ref{E11}). That is, using our
tontine fund language, that nobody survives to the end of the period. This is
the rationale for introducing the tontine administrator -- although in the
context of this appendix, this entity would be better described as the DRS
administrator. As far as the notation is concerned, the administrator would be
captured by $X_{n+1}$, being mutually exclusive, with $\sum_{j=1}^{n}X_{j}$.
In this case, $W_{i}^{\text{prop}}\left(  \mathbf{X}\right)  $ is always
well-defined, as the denominator in (\ref{E11}) is always strictly positive.

The tontine fund described in this paper arises as a special case of the
compensation-based DRS scheme (\ref{E11}) by choosing
\[
\mathbf{X}=\left(  f_{1}\times I_{1},\text{ }f_{2}\times I_{2},\ldots
,f_{n+1}\times I_{n+1}\right)  ,
\]
with the survival indicator and share allocation vectors $\mathbf{I}$ and
$\mathbf{f}$ as defined before in this paper.\ In this case, (\ref{E11})
transforms in (\ref{MD5}).

\subsection{DRS via Contributions}

In contrast to a system with premiums replaced by compensations, this section
considers a \emph{contribution}-based system for DRS. Consider again a pool of
$n+1$ individuals with random end-of-period claims $X_{i}$, properly
contracted and related to a well-defined peril or contingent benefit. The
total or combined vector of all dollar value claims is simply called the
\underline{\textbf{claims vector}}$\ $and is again denoted by the bold
$\mathbf{X}$:%
\[
\mathbf{X}=\left(  X_{1},X_{2},\ldots,X_{n+1}\right)
\]

As alluded earlier, DRS refers to risk-sharing mechanisms under which the
participants in the pool share their risks without generating or creating any
insolvency risk. To achieve that objective, in this particular section, we
assume that each of the $n+1$ participants in the risk-sharing pool is fully
compensated \textit{ex-post} for his claim $X_{i}$. But, in return each
participant pays an \emph{ex-post} contribution $C_{i}\left(  \mathbf{X}%
\right)  $ to the pool, which is managed by the insurance \emph{community}
versus a conventional insurance \emph{company.} The vector of all
contributions or allotments is called the \underline{\textbf{Contribution
Vector}}:%

\[
\mathbf{C}\left(  \mathbf{X}\right)  =\left(  C_{1}\left(  \mathbf{X}\right)
,C_{2}\left(  \mathbf{X}\right)  ,\ldots,C_{n+1}\left(  \mathbf{X}\right)
\right)
\]

Now, the function which transforms or maps any $\mathbf{X}$ into the
contribution vector $\mathbf{C}\left(  \mathbf{X}\right)  $ is called the
contribution-based risk-sharing rule. Note that in this formulation, we do not
necessarily assume that $\mathbf{C}$ is a function from $\mathbb{R}^{n}$ to
$\mathbb{R}^{n}$. Rather, the only assumption made here is that $\mathbf{C}$
is a function from an appropriate set of $n$-dimensional random vectors
defined on a given probability space, mapped to this same (or another) set of
random vectors. We denote the above-described contribution-based DRS scheme by
$\left(  \mathbf{X},\mathbf{C}\left(  \mathbf{X}\right)  \right)  $

Moving on, to avoid insolvency risk -- since there is no guarantor company --
which is present in centralized risk-sharing, we assume that the the
risk-sharing rule is such that the following condition, which we will call the
contribution-based solvency condition is fulfilled:%

\begin{equation}
\text{\textbf{Contribution-Based Solvency Condition}: }\sum_{j=1}^{n+1}%
X_{j}=\sum_{j=1}^{n+1}C_{j}\left(  \mathbf{X}\right)  , \label{E6b}%
\end{equation}

This condition is interpreted to mean that the sum of all contributions paid
by the participants (including the administrator) matches the sum of all
losses the pool covers. It should be clear that the cash-inflow for
participant $i$ at time $T=1$ in the above-described decentralized
risk-sharing scheme $\left(  \mathbf{X},\mathbf{C}\left(  \mathbf{X}\right)
\right)  $\ is given by:
\begin{equation}
X_{i}-C_{i}\left(  \mathbf{X}\right)  , \label{E7b}%
\end{equation}
which obviously may be positive or negative depending on the magnitude of
their loss relative to their share of the entire pool's losses.

For example, consider the \underline{\textbf{uniform risk-sharing rule}}
$\mathbf{C}^{\text{uni}}$, which is defined by
\begin{equation}
\mathbf{C}^{\text{uni}}\left(  \mathbf{X}\right)  =\left(  \frac{\sum
_{j=1}^{n+1}X_{j}}{n+1},\frac{\sum_{j=1}^{n+1}X_{j}}{n+1},\ldots,\frac
{\sum_{j=1}^{n+1}X_{j}}{n+1}\right)  \label{E8b}%
\end{equation}

Another example is the \underline{\textbf{conditional mean risk-sharing rule}}
$\mathbf{C}^{\text{cm}}$, introduced in the actuarial literature in Denuit \&
Dhaene (2012):%
\[
\mathbf{C}^{\text{cm}}\left(  \mathbf{X}\right)  =\left(  \mathbb{E}\left[
X_{1}\mid\sum_{j=1}^{n+1}X_{j}\right]  ,\mathbb{E}\left[  X_{2}\mid\sum
_{j=1}^{n+1}X_{j}\right]  ,\ldots,\mathbb{E}\left[  X_{n+1}\mid\sum
_{j=1}^{n+1}X_{j}\right]  \right)
\]
The properties of these and other risk-sharing rules have been investigated in
detail by Denuit, Dhaene \& Robert (2022) and Denuit, Dhaene, Ghossoub \&
Robert (2023) among others. An axiomatic characterization of the conditional
mean risk-sharing rule is given in Jiao, Kou, Liu \& Wang (2023){\small . }

\subsection{Some dualities}

In the subsections above, we have described two types of decentralized
risk-sharing; compensation-based and contribution-based systems, denoted by
$\left(  \boldsymbol{\pi},\mathbf{W}\left(  \mathbf{X}\right)  \right)  $\ and
$\left(  \mathbf{X},\mathbf{C}\left(  \mathbf{X}\right)  \right)  $,
respectively. Full solvency could be satisfied via appropriate time-$1$
compensations $\mathbf{W}\left(  \mathbf{X}\right)  $\ satisfying the
compensation-based solvency condition (\ref{E9}), or via appropriate time-$1$
contributions $\mathbf{C}\left(  \mathbf{X}\right)  $ satisfying the
contribution-based solvency condition in equation (\ref{E6b}).

Now, it is easy to see that in fact the two systems are 'equivalent'. Indeed,
the at time $1$ evaluated cash inflows for participant $i$ are equal under the
compensation-based system $\left(  \boldsymbol{\pi},\mathbf{W}\left(
\mathbf{X}\right)  \right)  $ and the contribution-based system $\left(
\mathbf{X},\mathbf{C}\left(  \mathbf{X}\right)  \right)  $ in case the
following conditions hold:%
\begin{equation}
W_{i}\left(  \mathbf{X}\right)  -\left(  1+R\right)  \times\pi_{i}=X_{i}%
-C_{i}\left(  \mathbf{X}\right)  ,\qquad i=1,2,\ldots,n+1. \label{E12}%
\end{equation}

This means that a contribution-based system $\left(  \mathbf{X},\mathbf{C}%
\left(  \mathbf{X}\right)  \right)  $ can be transformed into a
compensation-based system $\left(  \boldsymbol{\pi},\mathbf{W}\left(
\mathbf{X}\right)  \right)  $ where any person $i$ pays an arbitrary chosen
premiums $\pi_{i}$ at time $0$ and receives the following contribution at time
$1$:
\begin{equation}
W_{i}\left(  \mathbf{X}\right)  =\left(  1+R\right)  \times\pi_{i}+X_{i}%
-C_{i}\left(  \mathbf{X}\right)  ,\qquad i=1,2,\ldots,n+1. \label{E13}%
\end{equation}

Remark that transforming a contribution-based system into a compensation-based
system requires the choice of a premium vector.\ But the net payoff to the
participant is indifferent to the choice of this premium vector.\ 

Also remark that 'equivalence' has to be interpreted not in terms of equal
cashflows, but in terms of equal time-$1$ values of the total cashflows, which
is exactly expressed in the equations (\ref{E12}).

In a similar way, a compensation-based risk-sharing system $\left(
\boldsymbol{\pi},\mathbf{W}\left(  \mathbf{X}\right)  \right)  $ can be
transformed into a contribution-based risk-sharing system $\left(
\mathbf{X},\mathbf{C}\left(  \mathbf{X}\right)  \right)  $ where person $i$
receives $X_{i}$ at time $1$ and pays the following contribution at that
time:
\begin{equation}
C_{i}\left(  \mathbf{X}\right)  =\left(  1+R\right)  \times\pi_{i}+X_{i}%
-W_{i}\left(  \mathbf{X}\right)  ,\qquad i=1,2,\ldots,n. \label{E14}%
\end{equation}

As a first example, consider the contribution-based uniform risk-sharing
rule$\left(  \mathbf{X,C}^{\text{uni}}\left(  \mathbf{X}\right)  \right)  $
defined in (\ref{E8}).\ In this case, we have that
\[
C_{i}^{\text{uni}}\left(  \mathbf{X}\right)  =\frac{\sum_{j=1}^{n+1}X_{j}%
}{n+1},
\]
meaning that participant $i$ receives $X_{i}$ at time $1$, and at that same
time pays the contribution $\frac{\sum_{j=1}^{n+1}X_{j}}{n+1}$.\ Chosing a
premium vector $\boldsymbol{\pi}$\ and taking into account (\ref{E13}), we
transform $\left(  \mathbf{X,C}^{\text{uni}}\left(  \mathbf{X}\right)
\right)  $ into the compensation-based system $\left(  \boldsymbol{\pi
},\mathbf{W}\left(  \mathbf{X}\right)  \right)  $, with%
\[
W_{i}^{\text{uni}}\left(  \mathbf{X}\right)  =\left(  1+R\right)  \times
\pi_{i}+X_{i}-\frac{\sum_{j=1}^{n+1}X_{j}}{n+1}\mathbf{.}%
\]
In this system, participant $i$ pays the premium $\pi_{i}$ at time $0$, while
he receives the compensation $\left(  1+R\right)  \times\pi_{i}+X_{i}%
-\frac{\sum_{j=1}^{n+1}X_{j}}{n+1}$ at time $1$.

As a second example, consider the compensation-based risk-sharing scheme
$\left(  \boldsymbol{\pi},\mathbf{W}^{\text{prop}}\left(  \mathbf{X}\right)
\right)  $ defined in (\ref{E11}).\ We can transform this scheme in the
contribution-based scheme $\left(  \mathbf{X,C}^{\text{prop}}\left(
\mathbf{X}\right)  \right)  $, with%
\[
C_{i}^{\text{prop}}\left(  \mathbf{X}\right)  =\left(  1+R\right)  \times
\pi_{i}+X_{i}-\left(  1+R\right)  \times\left(  \sum_{j=1}^{n+1}\pi
_{j}\right)  \times\frac{X_{i}}{\sum_{j=1}^{n+1}X_{j}},\qquad\text{for
}i=1,2,\ldots,n+1
\]
In this system, at time $1$ participant $i$ receives $X_{i}$ and contributes
$C_{i}^{\text{prop}}\left(  \mathbf{X}\right)  $ to the insurance community.\

\end{document}